\documentclass[12pt]{amsart}
\usepackage{amssymb}
\usepackage{color}
\usepackage{amsmath,epic,curves,amscd}
\usepackage[english]{babel}
\usepackage{graphicx}
\usepackage{comment}
\usepackage{appendix}
\usepackage{mathdots}
\usepackage[all]{xy}
\usepackage{algorithm}
\usepackage{algpseudocode}

\newtheorem{claim}{}[section]
\newtheorem{theorem}[claim]{Theorem}

\newtheorem{lemma}[claim]{Lemma}
\newtheorem{proposition}[claim]{Proposition}
\newtheorem{corollary}[claim]{Corollary}

\theoremstyle{remark}

\renewenvironment{proof}{\noindent{\it Proof. \hskip0pt}}
                      {$\square$\par\medskip}

\newcommand\xx{{\text{\sf X}}}
\newcommand\lan{\langle}
\newcommand\ran{\rangle}

\newcommand\ot{\otimes}

\newcommand\image{{\text{\rm Im}}\,}

\newcommand\inte{{\text{\rm int}}\,}
\newcommand\ttt{{\text{\rm t}}}

\newcommand\re{{\text{\rm Re}}\,}

\newcommand\bfi{{\bf i}}
\newcommand\bfj{{\bf j}}

\newcommand\sa{{\rm sa}}
\newcommand\ph{{\rm ph}}

\newcommand\res{{\text{\rm res}}}

\newcommand\calv{{\mathcal V}}
\newcommand\calg{{\mathcal G}}
\newcommand\calt{{\mathcal T}}
\newcommand\calvnR{{\mathcal V}_n^{\mathbb R}}

\begin{document}
\baselineskip 6.0 truemm
\parindent 1.5 true pc

\title{Separability of multi-qubit states in terms of diagonal and anti-diagonal entries}

\begin{abstract}
We give separability criteria for general multi-qubit states in
terms of diagonal and anti-diagonal entries. We define two numbers
which are obtained from diagonal and anti-diagonal entries,
respectively, and compare them to get criteria. They give rise to
characterizations of separability when all the entries are zero
except for diagonal and anti-diagonal, like
Greenberger-Horne-Zeilinger diagonal states. The criteria is strong
enough to detect nonzero volume of entanglement with positive
partial transposes.
\end{abstract}

\author{Kil-Chan Ha, Kyung Hoon Han and Seung-Hyeok Kye}
\address{Kil-Chan Ha, Faculty of Mathematics and Statistics, Sejong University, Seoul 143-747, Korea}
\email{kcha at sejong.ac.kr}
\address{Kyung Hoon Han, Department of Data Science, The University of Suwon, Gyeonggi-do 445-743, Korea}
\email{kyunghoon.han at gmail.com}
\address{Seung-Hyeok Kye, Department of Mathematics and Institute of Mathematics, Seoul National University, Seoul 151-742, Korea}
\email{kye at snu.ac.kr}
\thanks{Both KHH and SHK were partially supported by NRF-2017R1A2B4006655, Korea}
\thanks{KCH was partially supported by NRF-2016R1D1A1A09916730, Korea}



\maketitle

\section{Introduction}

Entanglement is one of the key resources in the current quantum
information and computation theory, and it is very important to
distinguish entanglement among separability. Positivity of partial
transposes is one of the earlier separability criteria: Every
separable state is of PPT \cite{{choi-ppt},{peres}}. The converse is
true for $2\otimes 2$ and $2\otimes 3$ systems
\cite{{horo-1},{stormer},{woronowicz}}, and so we may confirm the
separability in those systems without a decomposition into the sums
of pure product states, which is just the definition of (full)
separability. A non-separable state is called entangled.

In this paper, we propose a separability criterion for multi-qubit systems by an inequality
between two numbers arising from the diagonal and anti-diagonal entries, respectively. This criterion
is also sufficient for separability when states have nonzero entries only for diagonal and anti-diagonal parts.
This is one of very few criteria in the literature with which we can confirm the separability without decomposition.

A state is called an {\sf X}-shaped state, or an {\sf X}-state if all the entries are zero
except for diagonal and anti-diagonal entries. Multi-qubit {\sf X}-states appear in various contexts
in quantum information theory \cite{{dur},{mendo},{rau},{vin10},{wein10},{yu}}.
Greenberger-Horne-Zeilinger diagonal states \cite{{bouw},{GHSZ},{GHZ}} are typical examples of
multi-qubit {\sf X}-states which have real anti-diagonals. Our criterion characterizes the separability
of GHZ diagonal state.

Following the earlier work
\cite{{guhne_pla_2011},{guhne10},{kay_2011}} on the separability of
three qubit GHZ diagonal states, the second and the third authors
\cite{han_kye_GHZ} characterized separability of three qubit GHZ
diagonal states. After they \cite{han_kye_phase} noticed that the
phases, the angular parts, of anti-diagonal entries play important
roles, a complete characterization \cite{chen_han_kye} of
separability of three qubit {\sf X}-states has been obtained. In
this paper, we explore analogues for general multi-qubit systems. The
main tool will be the duality \cite{kye_multi_dual} between
separability of multi-partite states and positivity of multi-linear
maps. In the next section, we fix the notations we use and summarize
the results in this paper.

The first and the second authors are grateful to Sanbeot community for their hospitality during their stay at the community hall.

\section{Notations and Summary of the results}

A function from the set $[n]=\{1,2,\dots,n\}$ into $\{0,1\}$ will be
called an $n$-bit index, just an $n$-index or index which will be
denoted by a string of $0,1$. The set of all $n$-bit indices will be
denoted by $I_{[n]}$. For example, we have $I_{[1]}=\{0,1\}$,
$I_{[2]}=\{00, 01,10,11\}$ and
$I_{[3]}=\{000,001,010,011,100,101,110,111\}$.  For a subset
$S\subset\{1,2,\dots,n\}$ and an index ${\bf i}\in I_{[n]}$, we
define the index $\bar{\bf i}^S\in I_{[n]}$ by
$$
\bar{\bf i}^S(k)=
\begin{cases}{\bf i}(k)+1\mod 2,\qquad &k\in S,\\{\bf i}(k),\qquad &k\notin S.\end{cases}
$$
In short, $\bar{\bf i}^S$ is obtained by switching the $k$-th
digit for $k\in S$. When $S=\{1,\dots,n\}$ is the whole set,
the index $\bar{\bf i}^{\{1,\dots,n\}}$ will be denoted by just $\bar{\bf i}$.

We denote by $\calvnR$ (respectively $\calv_n^+$) the set of all real
valued (respectively nonnegative) functions on the set $I_{[n]}$, and
by $\calv_n^\sa$ the set of all complex functions $u:\bfi\mapsto u_\bfi$ on $I_{[n]}$ satisfying the
relation $u_{\bar {\bf i}}=\bar u_{\bf i}$ for each $\bfi\in I_{[n]}$. We note that both
$\calvnR$ and $\calv_n^\sa$ are real vector spaces of dimension
$2^n$. For $s\in \calv_n^+$ and $u\in \calv_n^\sa$, we denote by
$X(s,u)$ the $n$-qubit self-adjoint matrix in the $n$-fold tensor product $M_2^{\otimes n}$ with
the entries $[w_{{\bf i},{\bf j}}]$ given by
$$
w_{{\bf i},{\bf j}}=\begin{cases}s_{\bf i},\qquad &{\bf j}={\bf i},\\
                     u_{\bf i},\qquad &{\bf j}=\bar {\bf i},\\
                     0,\qquad &{\text{\rm otherwise}}.
\end{cases}
$$
If we endow $I_{[n]}$ with the lexicographic order then $X(s,u)$ can be considered as a usual matrix
as follows:
$$
X(s,u)=\left(
\begin{matrix}
s_{00\dots 0} &&&&&&&&& u_{00\dots 0}\\
& \ddots &&&&&&& \iddots &\\
&& s_{\bfi} &&&&& u_{\bfi} & \\
&&& \ddots &&& \iddots &&\\
&&&& s_{01\dots 1}&u_{{01\dots 1}} &&&\\
&&&&  u_{10\dots 0}&s_{10\dots 0}&&&\\
&&& \iddots &&& \ddots &&\\
&&  u_{\bar\bfi} &&&&& s_{\bar\bfi} &\\
& \iddots &&&&&&& \ddots &\\
 u_{11\dots 1} &&&&&&&&& s_{11\dots 1}
\end{matrix}
\right).
$$

We first determine pairs $(s,u)\in \calv_n^+\times \calv_n^\sa$ for which
$W=X(s,u)$ is an entanglement witness. To do this, we introduce the notation:
$$
z^{\bf i}=\prod_{k=1}^n z_k^{1-2\bfi(k)},
$$
for ${\bf i}\in I_{[n]}$ and $z=(z_1,z_2,\dots,z_n)\in\mathbb C^n$
with $z_k\neq 0$ for $k=1,2,\dots,n$. We have $(e^{{\rm i}\theta_1}, e^{{\rm i}\theta_2}, e^{{\rm
i}\theta_3})^{001} =e^{{\rm i}\theta_1} e^{{\rm i}\theta_2} e^{-{\rm
i}\theta_3}$, for example. With this notation, we define the
following two numbers:
\begin{equation}\label{A_nB_n}
\delta_n(s)=\inf_{r\in\mathbb R_+^n}\sum_{{\bf i}\in{I}_{[n]}} s_{\bf i}r^{\bf i},\qquad
\|u\|_{\xx_n}=\sup_{\alpha\in\mathbb T^n}\sum_{{\bf i}\in{I}_{[n]}} u_{\bf i}\alpha^{\bf i},
\end{equation}
for $s\in \calv_n^+$ and $u\in \calv_n^\sa$,
where $\mathbb R_+=(0,\infty)$ and $\mathbb T$ is the unit circle on the complex plane.
Note that $u\mapsto \|u\|_{\xx_n}$ defines a norm on the real vector space $\calv_n^\sa$,
because the set $\{ \sum_{{\bf i}\in{I}_{[n]}} u_{\bf i}\alpha^{\bf i}\}$ of real numbers is
symmetric with respect to the origin. As for $\delta_n(s)$, we have the relation $\delta_n(\lambda s)=\lambda \delta_n(s)$ for $\lambda\ge 0$.

We show in Theorem \ref{block-pos} that a non-positive matrix $W=X(s,u)\in M_2^{\otimes n}$ is an entanglement witness if and only if
the inequality
$$
\delta_n(s)\ge \|u\|_{\xx_n}
$$
holds. In order to characterize the separability of an $n$-qubit state $\varrho=X(a,c)$,
we also introduce in Section \ref{sec-sep} the numbers
\begin{equation}\label{Delta_dual}
\begin{aligned}
\Delta_n(a)&=\inf \{\lan s,a\ran: s\in \calv_n^+,\ \delta_n(s)=1\},\qquad\\
\|c\|^\prime_{\xx_n}
&=\sup \{\lan u,c\ran : u\in \calv_n^\sa,\ \|u\|_{\xx_n}=1\},
\end{aligned}
\end{equation}
for $a\in\calv_n^+$ and $c\in\calv_n^\sa$, where $\lan s,a\ran=\sum_{\bfi\in I_{[n]}}s_\bfi a_\bfi$ is the usual
bi-linear pairing. Note that $\|\cdot \|^\prime_{\xx_n}$ is nothing but the dual norm of
$\|\cdot\|_{\xx_n}$. The main result in this paper, Theorem \ref{x-sep-th}, tells us that the state $\varrho=X(a,c)$ is separable if and only if
the inequality
$$
\Delta_n(a)\ge \|c\|^\prime_{\xx_n}
$$
holds. We show that the {\sf X}-part of a separable multi-qubit state is again separable,
and so the above inequality gives rise to a separability criterion for general $n$-qubit states
in terms of diagonal and anti-diagonal entries.

In the remainder of this paper, we estimate the numbers $\Delta_n(a)$ and $\|c\|^\prime_{\xx_n}$
to get separability criteria for multi-qubit states.
Recall that a multiset is a collection which allows repetition of elements, unlike a set.
A multiset $T$ of $n$-indices will be said to be {\sl balanced} if
$r\mapsto \prod_{{\bf i}\in{T}}r^{\bf i}$ is the constant function $1$ on $\mathbb R^n_+$. This happens if and only if
the number of indices ${\bf i}$ in $T$ with ${\bf i}(k)=0$ coincides with
the number of indices ${\bf i}\in{T}$ with ${\bf i}(k)=1$
for every $k=1,2,\dots,n$.
The cardinality of $T$ will be called the {\sl order} of ${T}$ and denoted by $\#{T}$.
The order of a balanced multiset must be even.
If $\#{T}=2$ then $\{{\bf i},{\bf j}\}$ is a balanced multiset
if and only if $\bar {\bf j}={\bf i}$. We show in Section \ref{sec-multiset} that the inequality
$\left(\prod_{\bfi\in  T}a_\bfi\right)^{1/\#(T)}\ge \Delta_n(a)$
holds for every balanced multiset $T$.
We say that a balanced multiset is {\sl irreducible}
when it cannot be partitioned into balanced multisets.
It is easily seen that the set $\calg_n$ of all irreducible balanced multisets of $n$-indices
is finite. We get in Section \ref{sec-multiset} the following upper bound
\begin{equation}\label{def_tilde_Delat}
\tilde\Delta_n(a):=\min\left\{ \left(\prod_{\bfi\in \bf T}a_\bfi\right)^{1/\#(T)} : T\in\calg_n\right\}\ge \Delta_n(a)
\end{equation}
for $\Delta_n(a)$. The number $\tilde\Delta_3(a)$ for the three qubit case appears
in G\"uhne's separability criterion \cite{guhne_pla_2011}. We show that the equality $\Delta_3(a)=\tilde\Delta_3(a)$ holds
for three qubit case, which recovers the main result in \cite{chen_han_kye}.
The notion of multisets is also useful to deal with the anti-diagonal part, and will be used to
characterize separability of half rank states $X(a,c)$ in Theorem \ref{sep-half-rank}. Especially, we show that
$\prod_{{\bfi}\in T}c_\bfi$ is constant for every irreducible balanced multiset $T$ of order four.
The number of irreducible balanced multisets of order four increases very rapidly as the number of qubits increases.

For a given $c\in\calv_n^\sa$ with the polar decompositions $c_\bfi=|c_\bfi|e^{{\rm i}\theta_\bfi}$, the function
$\theta\in\calvnR$ is called the {\sl phase part} of $c\in\calv_n^\sa$. Note that the phase part of a vector
in $\calv_n^\sa$ belongs to the subspace $\calv_n^\ph$ of $\calvnR$ consisting of all $\theta\in\calvnR$ satisfying
the relation $\theta_{\bar\bfi}=-\theta_\bfi$ for each $\bfi\in I_{[n]}$. Note that $\calv_n^\ph$ is of $2^{n-1}$
dimension. Section \ref{sec-phase} will be devoted to analyze the phase parts of the anti-diagonal entries of separable states.
The main tool is the linear map
${\Theta_n}:\mathbb R^n \to \calvnR$
defined by
\begin{equation}\label{theta}
[\Theta_n(e_k)]_\bfi=\begin{cases}+1,\quad &{\text{\rm if}}\ {\bf i}(k)=0,\\ -1, \quad &{\text{\rm if}}\ {\bf i}(k)=1,\end{cases}
\end{equation}
for each $k=1,2,\dots,n$, where $\{e_k\}$ is the usual orthonormal basis of $\mathbb R^n$.
We say that $\theta\in\calv_n^\ph$ {\sl satisfies the phase identities} when it belongs to the image of $\Theta_n$.
We construct a basis of the orthogonal complement $\calv_n^\ph\ominus\image\Theta_n$ arising from irreducible
balanced multisets, to express phase identities. In a circumstance, we will see that parts of the identities are
required for separability.

The {\sl phase difference} of $\theta\in\calv_n^\ph$ is defined
as the coset in the quotient $\calv_n^\ph/\image\Theta_n$ to which $\theta$ belongs. In Section \ref{sec-phase-diff},
we will show that the dual norm $\|c\|^\prime_{\xx_n}$ depends only on the phase difference of the phase part as well as the magnitudes.
We also show that $\|c\|^\prime_{\xx_n}$ is strictly greater than $\|c\|_\infty$ whenever $c\in\calv_n^\sa$ shares a
common magnitude of entries and has a nontrivial phase difference. From this,
we may construct boundary separable $n$-qubit states with full ranks for each $n\ge 3$.
This also tells us that our criterion is strong enough to detect PPT entanglement of nonzero volume.

\section{Multi-qubit {\sf X}-shaped entanglement witnesses}

We say that an $n$-qubit self-adjoint matrix $W$ in $\bigotimes_{k=1}^n M_2$ is {\sl block-positive}
when $\langle\varrho,W\rangle\ge 0$ for every separable state
$\varrho$. Note that a non-positive self-adjoint $W$ is an entanglement witness if and only if
it is block-positive. For a given partition $[n]=S\sqcup T$, we
defined in \cite{han_kye_optimal} the linear map $\phi_W^{S,T}$ from
$\bigotimes_{k\in S}M_2$ into $\bigotimes_{k\in T}M_2$, which is very useful to characterize the bi-separability
of multi-partite states. For a given $(n+1)$-qubit self-adjoint matrix $W=[w_{\bfi,\bfj}]_{{\bf i},{\bf j}\in{I}_{[n+1]}}$, we consider
the partition $[n+1]=\{n+1\}\sqcup [n]$ to get the map
\begin{equation}\label{phi-map}
\phi:=\phi^{\{n+1\},[n]}_W:M_2\to M_2^{\otimes n}.
\end{equation}
Following the construction in \cite{han_kye_optimal}, the map
$\phi$ sends $|i\rangle\langle j|\in M_2$ to $[w_{{\bf i}i,{\bf j}j}]_{{\bf i},{\bf j}\in{I}_{[n]}}\in M_2^{\otimes n}$,
and so $\lan\phi(|i\ran\lan j|),|{\bf i}\ran\lan {\bf j}|\ran=w_{{\bf i}i,{\bf j}j}$
for every ${\bf i},{\bf j}\in{I}_{[n]}$ and $i,j=0,1$.
Therefore, we have
$$
\lan W,|{\bf i}\ran\lan {\bf j}|\otimes |i\ran\lan j|\ran
=\lan W, |{\bf i}i\ran\lan {\bf j}j|\ran=w_{{\bf i}i,{\bf j}j}=\lan\phi(|i\ran\lan j|),|{\bf i}\ran\lan {\bf j}|\ran,
$$
which implies the following identity
$$
\lan\phi(a_{n+1}),a_1\otimes \cdots\otimes a_n\ran=
\lan W,a_1\otimes \cdots \otimes a_n\otimes a_{n+1}\ran,
$$
for every $2\times 2$ matrices $a_1,\dots,a_{n+1}$. Hence, we have the following:

\begin{lemma}\label{BP-induction}
An $(n+1)$-qubit self-adjoint matrix $W\in M_2^{\otimes (n+1)}$ is block-positive if and only
if $\phi(P)\in M_2^{\otimes n}$ is block-positive for every positive $P\in M_2$.
\end{lemma}

If $W$ is an {\sf X}-shaped $(n+1)$-qubit self-adjoint matrix with entries $\{w_{{\bf i},{\bf j}}:{\bf i},{\bf j}\in{I}_{[n+1]}\}$,
then $\phi(|0\rangle\langle 0|)$ and $\phi(|1\rangle\langle 1|)$ are
diagonal matrices with diagonal entries
$\{ w_{{\bf i}0,{\bf i}0}:{\bf i}\in{I}_{[n]}\}$ and $\{ w_{{\bf i}1,{\bf i}1}:{\bf i}\in{I}_{[n]}\}$,
respectively. We also note that $\phi(|0\rangle\langle 1|)$ and $\phi(|1\rangle\langle 0|)$ are
anti-diagonal matrix with the anti-diagonal entries
$\{ w_{{\bf i}0,\bar {\bf i}1}:{\bf i}\in{I}_{[n]}\}$ and
$\{ w_{{\bf i}1,\bar {\bf i}0}:{\bf i}\in{I}_{[n]}\}$, respectively.
We write
$P_{r,\alpha}=\left(\begin{matrix}r&\alpha\\ \bar\alpha& r^{-1}\end{matrix}\right)$. Then
sums of nonnegative scalar multiples of $P_{r,\alpha}$ with $r\in\mathbb R_+$ and $\alpha\in\mathbb T$
make a dense subset for $2\times 2$ positive matrices.
Now, we are ready to prove the main result in this section. Recall the definitions of
$\delta_n(s)$ and $\|u\|_{\xx_n}$ in (\ref{A_nB_n}) for $s\in \calv_n^+$ and $u\in\calv_n^\sa$.

\begin{theorem}\label{block-pos}
An $n$-qubit {\sf X}-shaped self-adjoint matrix $W=X(s,u)$ is block-positive if and only if
the inequality $\delta_n(s)\ge \|u\|_{\xx_n}$ holds.
\end{theorem}

\begin{proof}
When $n=1$ with ${I}_{[1]}=\{0,1\}$, we have
$$
\begin{aligned}
\delta_1(s)&=\inf_{r_1\in\mathbb R_+}(s_0r_1+s_1r_1^{-1})=2\sqrt{s_0s_1},\\
\|u\|_{\xx_1}&=\sup_{\alpha_1\in\mathbb T}(u_0\alpha_1+u_1\alpha_1^{-1})=\sup_{{\alpha_1}\in\mathbb T}(u_0\alpha_1+\bar u_0\bar\alpha_1)=2|u_0|
=\|u\|_1,
\end{aligned}
$$
which show that the inequality $\delta_1(s)\ge \|u\|_{\xx_1}$ holds if and only if $W$ is positive.

In order to use finite induction, suppose that the statement holds for the $n$-qubit case.
For an $(n+1)$-qubit {\sf X}-shaped self-adjoint matrix $W$, we see that the map $\phi$ in (\ref{phi-map})
sends $P_{r_{n+1},\alpha_{n+1}}$ to the $n$-qubit {\sf X}-shaped matrix
whose ${\bf i}$-th diagonal and anti-diagonal entries are given by
$$
r_{n+1}s_{{\bf i}0}+r_{n+1}^{-1}s_{{\bf i}1}\qquad{\text{\rm and}}\qquad
\alpha_{n+1}u_{{\bf i}0}+\alpha_{n+1}^{-1}u_{{\bf i}1},
$$
respectively. Furthermore, we have
$$
\begin{aligned}
\inf_{r_{n+1}\in\mathbb R_+} \inf_{r\in\mathbb R_+^n} \sum_{{\bf i}\in{I}_{[n]}} (r_{n+1}s_{{\bf i}0}+r_{n+1}^{-1}s_{{\bf i}1})r^{\bf i}&=
\inf_{r_{n+1}\in\mathbb R_+} \inf_{r\in\mathbb R_+^n} \sum_{{\bf i}\in{I}_{[n]}} (s_{{\bf i}0}r^{\bf i}r_{n+1}+s_{{\bf i}1}r^{\bf i}r_{n+1}^{-1})\\
&=\inf_{r\in\mathbb R_+^{n+1}} \sum_{{\bf j}\in{I}_{[n+1]}}s_{\bf j}r^{\bf j}
=\delta_{n+1}(s),
\end{aligned}
$$
and
$$
\sup_{\alpha_{n+1}\in\mathbb T} \sup_{\alpha\in\mathbb T^n}
\sum_{{\bf i}\in{I}_{[n]}}(\alpha_{n+1}u_{{\bf i}0}+\alpha_{n+1}^{-1}u_{{\bf i}1})\alpha^{\bf i}=\|u\|_{\xx_{n+1}},
$$
similarly. Therefore, we see that $W$ is block-positive if and only
if $\phi(P_{r_{n+1},\alpha_{n+1}})\in M_2^{\otimes n}$ is
block-positive for every $r_{n+1}\in \mathbb R_+$ and
$\alpha_{n+1}\in\mathbb T$ by Lemma \ref{BP-induction} if and only
if $\delta_{n+1}(s)\ge \|u\|_{\xx_{n+1}}$ by the induction hypothesis.
\end{proof}

In the case of $n=2$, we have
$$
\begin{aligned}
\delta_2(s)&=\inf_{r_1,r_2\in\mathbb R_+}(s_{00}r_1r_2+s_{01}r_1r_2^{-1}+s_{10}r_1^{-1}r_2+s_{11}r_1^{-1}r_2^{-1})\\
&=\inf_{r_2\in\mathbb R_+}\left[\inf_{r_1\in\mathbb R_+}[(s_{00}r_2+s_{01}r_2^{-1})r_1+(s_{10}r_2+s_{11}r_2^{-1})r_1^{-1}]\right]\\
&=\inf_{r_2\in\mathbb R_+}2\sqrt{(s_{00}r_2+s_{01}r_2^{-1})(s_{10}r_2+s_{11}r_2^{-1})}\\
&=\inf_{r_2\in\mathbb R_+}2\sqrt{s_{00}s_{10}r_2^2+s_{01}s_{11}r_2^{-2}+s_{00}s_{11}+s_{01}s_{10}}\\
&=2\sqrt{2\sqrt{s_{00}s_{01}s_{10}s_{11}}+s_{00}s_{11}+s_{01}s_{10}}
=2(\sqrt{s_{00}s_{11}}+\sqrt{s_{01}s_{10}}\,),
\end{aligned}
$$
and
$$
\begin{aligned}
\|u\|_{\xx_2}&=\sup_{\alpha_1,\alpha_2\in\mathbb T}(u_{00}\alpha_1\alpha_2+u_{01}\alpha_1\bar\alpha_2
   +u_{10}\bar\alpha_1\alpha_2+u_{11}\bar\alpha_1\bar\alpha_2)\\
&=\sup_{\alpha_1,\alpha_2\in\mathbb T}(u_{00}\alpha_1\alpha_2+u_{01}\alpha_1\bar\alpha_2
   +\bar u_{01}\bar\alpha_1\alpha_2+\bar u_{00}\bar\alpha_1\bar\alpha_2)\\
&=2\sup_{\alpha_1,\alpha_2\in\mathbb T} \left(\re (u_{00}\alpha_1\alpha_2+u_{01}\alpha_1\bar\alpha_2)\right)
=2(|u_{00}|+|u_{01}|)=\|u\|_1,
\end{aligned}
$$
because we can take $\alpha_1,\alpha_2\in\mathbb T$ so that the last equality holds. Alternatively, we have
$$
\begin{aligned}
\|u\|_{\xx_2}
&=2\sup_{\alpha_2\in\mathbb T} \left[\sup_{\alpha_1\in\mathbb T} \left(\re (u_{00}\alpha_2+u_{01}\bar\alpha_2)\alpha_1\right)\right]\\
&=2\sup_{\alpha_2\in\mathbb T} |u_{00}\alpha_2+u_{01}\bar\alpha_2|
=2(|u_{00}|+|u_{01}|)=\|u\|_1.
\end{aligned}
$$

We proceed to find inductive formulae for $\delta_n(s)$ and $\|u\|_{\xx_n}$. For a function $s\in\calv_{n+1}^{\mathbb R}$
defined on ${I}_{[n+1]}={I}_{[n]}\times \{0,1\}$
and a real number $r\in\mathbb R_+$, we define $s[r]\in\calvnR$ by
\begin{equation}\label{ind_I_2}
s[r]_{\bf i}=s_{{\bf i}0}r+s_{{\bf i}1}r^{-1},\qquad {\bf i}\in{I}_{[n]}.
\end{equation}
By the relation
$$
\begin{aligned}
\sum_{{\bf i}\in{I}_{[n+1]}} s_{\bf i}\cdot (r_1,\dots,r_n,r_{n+1})^{\bf i}
&=\sum_{{\bf j}\in{I}_{[n]}}(s_{{\bf j}0}r_{n+1}+s_{{\bf j}1}r_{n+1}^{-1}) \cdot (r_1,\dots,r_n)^{\bf j}\\
&=\sum_{{\bf j}\in{I}_{[n]}}s[r_{n+1}]_{\bf j} \cdot (r_1,\dots,r_n)^{\bf j},
\end{aligned}
$$
we have
\begin{equation}\label{ind_A_n}
\begin{aligned}
\delta_{n+1}(s)
=\inf_{r_{n+1}\in\mathbb R_+}\left[\inf_{(r_1,\dots,r_n)\in\mathbb R^n_+}
   \sum_{{\bf i}\in{I}_{[n]}}s[r_{n+1}]_{\bf i}\cdot (r_1,\dots,r_n)^{\bf i}\right]
=\inf_{r\in\mathbb R_+} \delta_n(s[r]).
\end{aligned}
\end{equation}
In the case of $n=3$, we have
$$
\begin{aligned}
\delta_3(s)
&=\inf_{r\in\mathbb R_+}\delta_2(s[r])\\
&=2\inf_{r\in\mathbb R_+}\left[\sqrt{s[r]_{00}s[r]_{11}}+\sqrt{s[r]_{01}s[r]_{10}}\right]\\
&=2\inf_{r\in\mathbb R_+}\left[\sqrt{(s_{000}r+s_{001}r^{-1})(s_{110}r+s_{111}r^{-1})}\right. \\
&\phantom{xxxxxxxxxxxxxxxxxxxxx}    +\left.\sqrt{(s_{010}r+s_{011}r^{-1})(s_{100}r+s_{101}r^{-1})}\right],
\end{aligned}
$$
for $s\in\calv_3^+$.
We can also take the first-third or second-third indices for ${\bf i}$ in  (\ref{ind_I_2}), to get
$$
\begin{aligned}
\delta_3(s)
&=2\inf_{r\in\mathbb R_+}\left[\sqrt{(s_{000}r+s_{010}r^{-1})(s_{101}r+s_{111}r^{-1})}\right.\\
&\phantom{xxxxxxxxxxxxxxxxxxxxx}    +\left.\sqrt{(s_{001}r+s_{011}r^{-1})(s_{100}r+s_{110}r^{-1})}\right]\\
&=2\inf_{r\in\mathbb R_+}\left[\sqrt{(s_{000}r+s_{100}r^{-1})(s_{011}r+s_{111}r^{-1})}\right.\\
&\phantom{xxxxxxxxxxxxxxxxxxxxx}    +\left.\sqrt{(s_{001}r+s_{101}r^{-1})(s_{010}r+s_{110}r^{-1})}\right].
\end{aligned}
$$
The last one has been considered in \cite{han_kye_tri,han_kye_GHZ}
to characterize three qubit block-positivity up to the scalar multiplication by $2$.

For given $u\in\calv_{n+1}^\sa$ and $\alpha\in\mathbb T$, we define $u[\alpha]\in\calv_n^\sa$ by
$u[\alpha]_\bfi=u_{\bfi 0}\alpha+u_{\bfi 1}\bar\alpha$ for $\bfi\in I_{[n]}$. Then we have
$$
\|u\|_{\xx_{n+1}}=\sup_{\alpha\in\mathbb T}\|u[\alpha]\|_{\xx_n}
$$
by the same reasoning. We also have
$$
\begin{aligned}
\|u\|_{\xx_3}
&=2\sup_{\alpha\in\mathbb T}\left( |u_{000}\alpha +u_{001}\bar\alpha|+|u_{010}\alpha+u_{011}\bar\alpha|\right)\\
&=2\sup_{\alpha\in\mathbb T}\left( |u_{000}\alpha +u_{001}|+|u_{010}\alpha+u_{011}|\right),
\end{aligned}
$$
and  the identities
$$
\begin{aligned}
\|u\|_{\xx_3}
&=2\sup_{\alpha\in\mathbb T}\left( |u_{000}\alpha +u_{010}|+|u_{001}\alpha+u_{011}|\right)\\
&=2\sup_{\alpha\in\mathbb T}\left( |u_{000}\alpha +u_{100}|+|u_{001}\alpha+u_{101}|\right).
\end{aligned}
$$
Motivated by the characterization of block-positivity of three qubit {\sf X}-shaped matrices in
\cite{han_kye_tri}, the half of the last number was taken as the definition of $B(u)$
in \cite{han_kye_GHZ}, and has been calculated \cite{chen_han_kye, han_kye_GHZ} in terms of entries $u_{\bf i}$ in several cases.

\begin{proposition}\label{b_n_ineq}
We have $2\|u\|_\infty\le \|u\|_{\xx_n}\le \|u\|_1$ for every $u\in\calv_n^\sa$.
\end{proposition}

\begin{proof}
The inequality $\|u\|_{\xx_n}\le \|u\|_1$ follows from
\begin{equation}\label{bbbbbb}
\|u\|_{\xx_n}
=\max_{\alpha\in\mathbb T^n}\sum_{{\bf i}\in{I}_{[n]}} u_{\bf i}\alpha^{\bf i}
\le \max_{\alpha\in\mathbb T^n}\sum_{{\bf i}\in{I}_{[n]}} |u_{\bf i}\alpha^{\bf i}|
=\sum_{{\bf i}\in{I}_{[n]}} |u_{\bf i}| =\|u\|_1.
\end{equation}
We will use finite induction to prove
the other inequality. Given $\bfi\in I_{[n]}$, we take $\alpha\in\mathbb T$ such that
$|u_{\bfi 0}\alpha+u_{\bfi 1}\bar\alpha|=|u_{\bfi 0}|+|u_{\bfi 1}|$. Then we have
$$
\|u\|_{\xx_{n+1}}\ge \|u[\alpha]\|_{\xx_n}
\ge 2 |u[\alpha]_\bfi|
\ge 2\max\{ |u_{\bfi 0}|,|u_{\bfi 1}|\}.
$$
We used the induction hypothesis in the second inequality. This shows $\|u\|_{\xx_{n+1}}\ge 2|u_\bfj|$ for every
$\bfj\in I_{[n+1]}$.
\end{proof}

\section{Separability of {\sf X}-shaped multi-qubit states}\label{sec-sep}

In this section, we characterize the separability of multi-qubit
{\sf X}-states $\varrho=X(a,c)$. We begin with the following
observation whose three qubit version appears in \cite{han_kye_GHZ}.

\begin{proposition}\label{x-part-sep}
The {\sf X}-part of a multi-qubit separable state is again separable.
\end{proposition}

\begin{proof}
It suffices to prove that the {\sf X}-part of an $n$-qubit pure product state is separable.
For a given $|x\ran=(x_0,x_1)^\ttt\in\mathbb C^2$, we write $|x_\pm\ran =(x_0,\pm x_1)^\ttt$. When $n=2$ and
$\varrho=|x \ran\lan x| \otimes |y \ran\lan y|$, the {\sf X}-part $\varrho_\xx$ of
$\varrho$ is given by
$$
\varrho_X = {1 \over 2}\left(|x_+ \ran\lan x_+| \otimes |y_+
\ran\lan y_+| + |x_- \ran\lan x_-| \otimes |y_- \ran\lan
y_-|\right).
$$
We will proceed by induction. Let $\omega$ be an
$(n-1)$-qubit pure product state and $\varrho=\omega \otimes |x \ran\lan
x|$. By the induction hypothesis, the {\sf X}-part $\omega_X$ of
$\omega$ is separable. Let $\omega_X^-$ be a separable state
obtained by multiplying the anti-diagonal part of
$\omega_X$ by $-1$. This is obtained by the local unitary operation with ${\rm diag}
(1,-1) \otimes (1,1) \otimes \cdots \otimes (1,1)$.

Then, $\omega_X \otimes |x \ran\lan x|$ is a block {\sf X}-state whose blocks are $2\times 2$ matrices,
and its diagonal and the anti-diagonal blocks are given by
$$
\omega_{\bfi,\bfi} \begin{pmatrix} |x_0|^2 & x_0 \bar x_1 \\ \bar
x_0 x_1 & |x_1|^2 \end{pmatrix} \qquad \text{and} \qquad
\omega_{\bfi,\bar \bfi} \begin{pmatrix} |x_0|^2 & x_0 \bar x_1 \\
\bar x_0 x_1 & |x_1|^2 \end{pmatrix},
$$
respectively. Similarly, $\omega_X^- \otimes |x_- \ran\lan x_-|$ is
also a block {\sf X}-state whose diagonal and the
anti-diagonal blocks are given by
$$
\omega_{\bfi,\bfi} \begin{pmatrix} |x_0|^2 & -x_0 \bar x_1 \\ -\bar
x_0 x_1 & |x_1|^2 \end{pmatrix} \quad \text{and} \quad
-\omega_{\bfi,\bar \bfi} \begin{pmatrix} |x_0|^2 & -x_0 \bar x_1 \\
-\bar x_0 x_1 & |x_1|^2 \end{pmatrix} =\omega_{\bfi,\bar \bfi}
\begin{pmatrix} -|x_0|^2 & x_0 \bar x_1 \\ \bar x_0 x_1 & -|x_1|^2
\end{pmatrix},
$$
respectively. Therefore, we have
$$
\varrho_X = {1 \over 2}\left(\omega_X \otimes |x \ran\lan x| +
\omega_X^- \otimes |x_- \ran\lan x_-|\right),
$$
which is separable.
\end{proof}

By the exactly same argument as in Proposition 3.1 of \cite{han_kye_GHZ}, we have the following:

\begin{corollary}\label{x-part-BP}
The {\sf X}-part of a multi-qubit block-positive matrix is again block-positive.
\end{corollary}

\begin{proposition}\label{sep-equi}
For an $n$-qubit state $\varrho=X(a,c)$ with $a\in \calv_n^+$ and $c\in \calv_n^\sa$, the following are equivalent:
\begin{enumerate}
\item[(i)]
$\varrho$ is separable;
\item[(ii)]
$\lan W,\varrho\ran\ge 0$ for every {\sf X}-shaped block-positive $n$-qubit matrix $W=X(s,u)$;
\item[(iii)]
$\delta_n(s)\ge \|u\|_{\xx_n}$ with $s\in \calv_n^+$ and $u\in \calv_n^\sa$ implies $\lan a,s\ran+\lan c,u\ran\ge 0$;
\item[(iv)]
$\delta_n(s)\ge \|u\|_{\xx_n}$ with $s\in \calv_n^+$ and $u\in \calv_n^\sa$ implies $\lan a,s\ran\ge\lan c,u\ran$;
\item[(v)]
$\delta_n(s)= \|u\|_{\xx_n}$ with $s\in \calv_n^+$ and $u\in \calv_n^\sa$ implies $\lan a,s\ran\ge\lan c,u\ran$;
\item[(vi)]
$\displaystyle{
\inf_{s\in \calv_n^+}\frac{\lan a,s\ran}{\delta_n(s)}\ge\sup_{u\in \calv_n^\sa}\frac{\lan c,u\ran}{\|u\|_{\xx_n}}}$.
\end{enumerate}
\end{proposition}

\begin{proof}
The equivalence (i) $\Longleftrightarrow$ (ii) follows from the duality and Corollary \ref{x-part-BP}.
On the other hand, Theorem \ref{block-pos} tells us  that (ii) and (iii) are equivalent to each other,
because $\lan \varrho, X(s,u)\ran= \lan a,s\ran +\lan c,u\ran$.
We also get (iii) $\Longleftrightarrow$ (iv) by replacing $u$ by $-u$, and the direction (iv) $\Longrightarrow$ (v) is clear.
Now, we prove the direction (v) $\Longrightarrow$ (iv).
If $\delta_n(s)>\|u\|_{\xx_n}$ then we take $\lambda = \|u\|_{\xx_n}/\delta_n(s) \in [0,1)$ to get
$\delta_n(\lambda s)=\|u\|_{\xx_n}$. Therefore, we have $\lan a,s\ran\ge \lan a,\lambda s\ran\ge \lan c,u\ran$.
The remaining implications (v) $\Longleftrightarrow$ (vi) follows from $\delta_n({s / \delta_n(s)})=\|{u / \|u\|_{\xx_n}}\|_{\xx_n}=1$.
\end{proof}

Recall the definitions of $\Delta_n(a)$ and $\|c\|_{\xx_n}^\prime$ in (\ref{Delta_dual}). Then
the inequality in Proposition \ref{sep-equi} (vi) is nothing but
the following main result in this paper:

\begin{theorem}\label{x-sep-th}
An {\sf X}-shaped $n$-qubit state $\varrho=X(a,c)$ is separable if and only if
the inequality $\Delta_n(a)\ge \|c\|_{\xx_n}^\prime$ holds.
\end{theorem}

By Proposition \ref{x-part-sep}, we have the following criterion for general multi-qubit states.

\begin{theorem}\label{criterion}
Let $\varrho$ be a multi-qubit state with the {\sf X}-part $X(a,c)$. If $\varrho$ is separable then
the inequality $\Delta_n(a)\ge \|c\|_{\xx_n}^\prime$ holds.
\end{theorem}

By Proposition \ref{b_n_ineq} and the duality, we have
\begin{equation}\label{dual-norm-ineq}
\|c\|_\infty\le \|c\|^\prime_{\xx_n}\le\frac 12 \|c\|_1,
\end{equation}
for every $c\in\calv_n^\sa$.
We also have
$$
\|c\|^\prime_{\xx_1}=\|c\|_\infty,\qquad \|c\|^\prime_{\xx_2}=\|c\|_\infty,
$$
because $\|u\|_{\xx_1}=\|u\|_1$ and $\|u\|_{\xx_2}=\|u\|_1$.

In order to estimate $\Delta_n(a)$ for $a\in\calv_n^+$, we consider $s_\lambda\in \calv^+_n$ defined by
$$
s_\lambda=\frac 12\left(\lambda\, e_{\bf i}
+\lambda^{-1} e_{\bar {\bf i}}\right)\in\calv_n^+
$$
for each $\lambda>0$,
where $e_{\bfi}\in \calv_n^{\mathbb R}$ is given by $(e_\bfi)_\bfj=1$ for $\bfi=\bfj$ and $(e_\bfi)_\bfj=0$ for $\bfi\neq\bfj$.
Then it is easy to see that $\delta_n(s_\lambda)=1$. We also have
$$
\inf_{\lambda>0}\lan a,s_\lambda\ran=\inf_{\lambda>0}\frac 12\left(\lambda a_\bfi +\lambda^{-1} a_{\bar\bfi}\right)
=\sqrt{a_\bfi a_{\bar\bfi}}.
$$
Therefore, we see that
$\Delta_n(a)\le\sqrt{a_{\bf i} a_{\bar {\bf i}}}$ for each ${\bf i}\in{I}_{[n]}$, and we have
\begin{equation}\label{delta=esi}
\Delta_n(a)\le \min\{\sqrt{a_{\bf i} a_{\bar {\bf i}}}:{\bf i}\in{I}_{[n]}\},\qquad a\in \calv_n^+,\ n=1,2,\dots
\end{equation}
When $n=1$, we have
$$
a_0s_0+a_1s_1\ge \sqrt{a_0a_1}\cdot 2\sqrt{s_0s_1}
=\sqrt{a_0a_1}\, \delta_1(s),
$$
because $\delta_1(s)=2\sqrt{s_0s_1}$.
In the case of $n=2$, we also have
$$
\begin{aligned}
s_{00}a_{00}&+s_{01}a_{01}+s_{10}a_{10}+s_{11}a_{11}\\
&\ge 2\sqrt{s_{00}s_{11}a_{00}a_{11}}+2\sqrt{s_{01}s_{10}a_{01}a_{10}}\\
&\ge\min\{\sqrt{a_{00}a_{11}},\sqrt{a_{01}a_{10}}\}\cdot 2(\sqrt{s_{00}s_{11}}+\sqrt{s_{01}s_{10}})\\
&=\min\{\sqrt{a_{00}a_{11}},\sqrt{a_{01}a_{10}}\}\cdot \delta_2(s)
\end{aligned}
$$
because $\delta_2(s)=2(\sqrt{s_{00}s_{11}}+\sqrt{s_{01}s_{10}})$.
Therefore, we have the following:
$$
\Delta_1(a)=\sqrt{a_0a_1},\qquad
\Delta_2(a)=\min\{\sqrt{a_{00}a_{11}},\sqrt{a_{01}a_{10}}\},
$$
and the inequality $\Delta_2(a)\ge \|c\|^\prime_{\xx_2}$ says that the $2$-qubit state $\varrho=X(a,c)$
is of PPT, which is equivalent to the separability in this case.

By the relation $\delta_n(s)\le\|s\|_1$, we have $\min\{a_{\bf
i}\}\delta_n(s)\le \langle s,a\rangle$, which implies
$$
\min\{a_{\bf i}:{\bf i}\in{I}_{[n]}\}\le \Delta_n(a),\qquad a\in \calv_n^+.
$$
If $a_{\bf i}=a_{\bar{\bf i}}$ for every index ${\bf i}\in I_{[n]}$, then we
have $\Delta_n(a)=\min\{a_{\bf i}:{\bf i}\in{I}_{[n]}\}$ by
(\ref{delta=esi}). Hence, we have the following:

\begin{corollary}\label{GHZ}
Suppose that the diagonal entries of an $n$-qubit state
$\varrho=X(a,c)$ satisfies $a_{\bf i}=a_{\bar{\bf i}}$ for each
${\bf i}\in{I}_{[n]}$. Then $\varrho$ is separable if and only if
$\min_{\bf i}a_{\bf i}\ge \|c\|^\prime_{\xx_n}$.
\end{corollary}

GHZ diagonal states are typical examples of
$n$-qubit states satisfying the condition of Corollary \ref{GHZ}.
These are states which are diagonal in the $n$-qubit orthonormal
GHZ-basis \cite{bouw, GHZ} consisting of $2^n$ vectors given by
$$
|\xi_{\bf i}\rangle = \frac 1 {\sqrt2} \bigl (|{\bf i}\rangle
+(-1)^{i_1}|\bar{\bf i}\rangle \bigr), \qquad {\bf i}=i_1i_2\cdots i_n.
$$
So, every $n$-qubit GHZ diagonal state is of the form $X(a,c)$ with
$a_{\bf i}=a_{\bar {\bf i}} \ge 0$ and $c_{\bf i}=c_{\bar
{\bf i}} \in \mathbb R$ for each ${\bf i} \in I_n$. In the case of $3$-qubit GHZ diagonal states, the dual norm
$\|c\|^\prime_{\xx_n}$ has been calculated in terms of anti-diagonal entries \cite{han_kye_GHZ}, which are real numbers.

\section{Separability and multiset of indices}\label{sec-multiset}

In order to improve the inequality (\ref{delta=esi}),
we consider multisets of indices. Recall that a multiset of $n$-indices is called balanced if
the following
$$
\#\{\bfi\in T:\bfi(k)=0\}=\#\{\bfi\in T:\bfi(k)=1\}
$$
holds for every $k=1,2,\dots,n$.

\begin{proposition}\label{nece-cri}
For every $a\in \calv_n^+$, the inequality
$$
\left(\prod_{{\bf i}\in{T}}a_{\bf i}\right)^{1/\ell}\ge \Delta_n(a)
$$
holds for every balanced multiset ${T}$ of length $\ell$.
\end{proposition}

\begin{proof}
If $a_\bfi=0$ for some $\bfi\in I_{[n]}$ then we have $\Delta_n(a)=0$ by (\ref{delta=esi}).
We consider the case when  $a_\bfi>0$ for every $\bfi\in I_{[n]}$.
For a given $a\in \calv_n^+$ and a balanced multiset ${T}$
with $\#{T}=\ell$, we define
$$
s=\frac 1{\ell}\left(\prod_{{\bf i}\in{T}}a_{\bf i}\right)^{1/\ell}
    \sum_{{\bf i}\in{T}}\frac 1{a_{\bf i}}e_{\bf i}\in \calv_n^+.
$$
Then we have
$$
\delta_n(s)=\frac 1{\ell}\left(\prod_{{\bf i}\in{T}}a_{\bf i}\right)^{1/\ell}
         \inf_{r\in \mathbb R^n_+}\left(\sum_{{\bf i}\in{T}}\frac 1{a_{\bf i}}r^{\bf i}\right)
\ge \frac 1{\ell}\left(\prod_{{\bf i}\in{T}}a_{\bf i}\right)^{1/\ell} \ell
         \left(\prod_{{\bf i}\in{T}}\frac 1{a_{\bf i}}\right)^{1/\ell}=1,
$$
and so, it follows that
$$
\Delta_n(a)\le\left\lan a,\frac s{\delta_n(s)}\right\ran
=\frac 1{\delta_n(s)}\lan a,s\ran \le \lan a,s\ran =\left(\prod_{{\bf i}\in{T}}a_{\bf i}\right)^{1/\ell},
$$
as it was desired.
\end{proof}

Proposition \ref{nece-cri} gives us nontrivial restrictions on the {\sf X}-parts $X(a,c)$ of separable multi-qubit states.
Especially, we have the restriction on the diagonal parts
$$
\left(\prod_{{\bf i}\in{T}}a_{\bf i}\right)^{1/\ell}\ge \|c\|_\infty
$$
by Theorem \ref{x-sep-th} and (\ref{dual-norm-ineq}), whenever ${T}$ is a balanced multiset of order $\ell$.

Suppose that a balanced multiset ${T}$ can be partitioned into
the multiset union of balanced multisets ${T}_1$ and
${T}_2$. Then we have the inequality
$$
\min\left\{
\left(\prod_{{\bf i}\in{T}_1}a_{\bf i}\right)^{1/{\# {T}_1}},\ \
\left(\prod_{{\bf i}\in{T}_2}a_{\bf i}\right)^{1/{\# {T}_2}} \right\}
\le \left(\prod_{{\bf i}\in{T}}a_{\bf i}\right)^{1/{\# {T}}},
$$
by taking logarithms. Therefore, we may consider only irreducible balanced multisets
when we estimate the number $\Delta_n(a)$ using Proposition \ref{nece-cri}.
Furthermore, there exist only finitely many irreducible balanced multisets of $n$-indices.
In fact, the maximum possible order of an irreducible balanced multiset is $2^{n-1}$, because
if the order of a balanced multiset $T$ exceeds $2^{n-1}$ then there exists $\bfi\in I_{[n]}$
such that both $\bfi$ and $\bar\bfi$ belong to $T$.
Therefore, we can define the number $\tilde\Delta_n(a)$ for $a\in\calv_n^+$ by (\ref{def_tilde_Delat}),
and get the following:

\begin{theorem}\label{Delat_tilte_Delat}
For every $a\in\calv_n^+$, we have
$\tilde\Delta_n(a)\ge \Delta_n(a)$.
\end{theorem}

The notion of balanced multisets is also useful to get separability criteria regarding anti-diagonal entries.
To see this, we consider the separability of $n$-qubit states whose {\sf X}-part is of rank $2^{n-1}$.
For $r\in\mathbb R_+^n$ and $\alpha\in\mathbb T^n$, we define $\tilde r\in \calv_n^+$ and $\tilde\alpha\in \calv_n^\sa $ by
$$
\tilde r_{\bf i}=r^{\bf i}\quad {\text{\rm and}}\quad \tilde\alpha_{\bf i}=\alpha^{\bf i},\qquad \bfi\in I_{[n]},
$$
respectively. If $\delta_n(s)\ge \|u\|_{\xx_n}$, then we have
$$
\langle \tilde r, s \rangle = \sum_{{\bf i}\in{I}_{[n]}} r^\bfi s_\bfi \ge
\delta_n(s) \ge \|u\|_{\xx_n} \ge \sum_{{\bf i}\in{I}_{[n]}} \alpha^\bfi u_\bfi = \langle \tilde \alpha, u \rangle,
$$
so the {\sf X}-state $\varrho=X(\tilde r,\tilde\alpha)$ is separable by
Proposition \ref{sep-equi}. Furthermore,
it is of half rank $2^{n-1}$. We show that every separable {\sf X}-state of half rank is
in this form up to positive scalar multiples.
Suppose that $\varrho=X(a,c)$ is separable.  Then we have $a_{\bf
i}a_{\bar{\bf i}}\ge |c_{\bf j}|^2$ for every ${\bf i}, {\bf
j}\in{I}_{[n]}$ by the PPT condition in \cite{han_kye_optimal}. If a
separable state $\varrho=X(a,c)$ is of half rank $2^{n-1}$ then the
identity  $a_{\bf i}a_{\bar{\bf i}}= |c_{\bf j}|^2$ holds for every
${\bf i}, {\bf j}\in{I}_{[n]}$. Without loss of generality, we may
assume that $a_{\bf i}a_{\bar{\bf i}}= |c_{\bf i}|^2=1$ for every
index ${\bf i}$, to characterize the separability for multi qubit
{\sf X}-states with half rank.

\begin{theorem}\label{sep-half-rank}
Let $\varrho=X(a,c)$ be an $n$-qubit {\sf X}-state with $a_{\bf i}a_{\bar{\bf i}}= |c_{\bf i}|^2=1$ for every index ${\bf i}$.
Then the following are equivalent:
\begin{enumerate}
\item[(i)]
$\varrho$ is separable;
\item[(ii)]
there exist $r\in\mathbb R_+^n$ and $\alpha\in\mathbb T^n$ such that $a=\tilde r$ and $c=\tilde\alpha$;
\item[(iii)]
there exists a product vector $|\xi\ran$ such that $\varrho$ is the {\sf X}-part of $|\xi\ran\lan\xi|$;
\item[(iv)]
$\prod_{{\bf i}\in{T}} a_{\bf i}= \prod_{{\bf i}\in{T}} c_{\bf i}=1$ for every balanced multiset ${T}$;
\item[(v)]
$\prod_{{\bf i}\in{T}} a_{\bf i}= \prod_{{\bf i}\in{T}} c_{\bf i}=1$ for every irreducible balanced multiset ${T}$;
\item[(vi)]
$\prod_{{\bf i}\in{T}} a_{\bf i}= \prod_{{\bf i}\in{T}} c_{\bf i}=1$ for every irreducible balanced multiset ${T}$ of order four.
\end{enumerate}
\end{theorem}

\begin{proof}
We first note that the separability of $\varrho$ implies
$$
1=\|c\|_\infty\le \|c\|^\prime_{\xx_n}\le \Delta_n(a)\le \min\{ \sqrt{a_{\bf i}a_{\bar{\bf i}}}:{\bf i}\in{I}_{[n]}\}=1,
$$
by (4), and so we have $\Delta_n(a)=\|c\|^\prime_{\xx_n}=1$. Taking
$s=\sum_{{\bf i}\in{I}_{[n]}}a_{\bar{\bf i}}e_{\bf i}\in \calv_n^+$, we have
$$
2^n=\lan a,s\ran \ge \Delta_n(a)\delta_n(s)=\delta_n(s)
=\inf_{r\in \mathbb R^n_+}\sum_{{\bf i}\in{I}_{[n]}} a_{\bar {\bf i}}r^{\bf i}
={1 \over 2} \inf_{r\in \mathbb R^n_+}\sum_{{\bf i}\in{I}_{[n]}} (a_{\bar {\bf i}}r^{\bf i} + a_{\bf i}r^{\bar \bfi})
\ge 2^n,
$$
since $a_{\bar {\bf i}}=a^{-1}_{\bf i}$ and $r^{\bar \bfi}= (r^{\bfi})^{-1}$.
Therefore, we see that the equality holds in the above inequality,
and so there exists $r\in \mathbb R^n_+$ such that $a_{\bar {\bf i}}r^{\bf i}=a_{\bf i}r^{\bar \bfi}$,
that is, $a_{\bf i}=r^{\bf i}$ for each ${\bf i}\in{I}_{[n]}$.
This tells us that $a=\tilde r$. On the other hand, we take
$u=\sum_{{\bf i}\in{I}_{[n]}}\bar c_{\bf i}e_{\bf i}\in \calv_n^\sa $ to get
$$
2^n=\lan c,u\ran\le\|c\|^\prime_{\xx_n} \|u\|_{\xx_n}=\|u\|_{\xx_n}
=\sup_{\alpha\in\mathbb T^n}\sum_{{\bf i}\in{I}_{[n]}} \bar c_{{\bf i}}\alpha^{\bf i}\le 2^n.
$$
This shows that there exists $\alpha\in\mathbb T^n$ such that $c=\tilde\alpha$ by the same argument, and we have the equivalence
(i) $\Longleftrightarrow$ (ii).

For $z\in\mathbb C^n$, we denote by $z^2$ the vector in $\mathbb C^n$ whose $k$-th entry is given by $z_k^2$.
For the direction (ii) $\Longrightarrow$ (iii), we
take $s\in\mathbb R^n_+$ and $\beta\in\mathbb T^n$ such that $s^2=r$ and $\beta^2=\alpha$, and consider the vector
\begin{equation}\label{vhftmnkmk}
|\xi\ran
=(s_1\beta_1, s_1^{-1}\bar\beta_1)^\ttt\otimes \dots\otimes (s_n\beta_n, s_n^{-1}\bar\beta_n)^\ttt.
\end{equation}
We note that the ${\bf i}$-th entry of $|\xi\ran$ is $s^{\bf i}\beta^{\bf i}$, and so we see that
the {\sf X}-part of $|\xi\ran\lan\xi|$ is just $X(\tilde r,\tilde\alpha)$. For the direction (iii) $\Longrightarrow$ (ii),
we note that every product vector $|\xi\ran$ with nonzero entries can be expressed by (\ref{vhftmnkmk}) up to scalar multiplications.
The implication (ii) $\Longrightarrow$ (iv) follows from the definition of balanced multisets,
and the implications  (iv) $\Longrightarrow$ (v)  $\Longrightarrow$ (vi) are trivial.

It remains to show the direction (vi) $\Longrightarrow$ (ii).
To do this, we first show that the system
\begin{equation}\label{equa-sep}
r^{\bf i}=a_{\bf i},\qquad {\bf i}\in{I}_{[n]}
\end{equation}
of equations with unknowns $r_1, r_2,\dots, r_n$ can be solved.
If we choose $\bfi\in I_{[n]}$ and $k \in [n]$ and
we multiply two equations $r^\bfi=a_\bfi$ and $r^{\bar \bfi^{\{k\}^c}}=a_{\bar \bfi^{\{k\}^c}}$,
then all but $r_k$ are canceled. So, we get a candidate for the solution as
$$
r_k =
\begin{cases}
\sqrt{a_\bfi a_{\bar \bfi^{\{k\}^c}}}, & \qquad \bfi(k)=0, \\
1 \slash \sqrt{a_\bfi a_{\bar \bfi^{\{k\}^c}}}, & \qquad \bfi(k)=1.
\end{cases}
$$
We show that  $r=(r_1,\dots,r_n)$ is independent of the choice of $\bfi \in I_{[n]}$.
Let $\bfi, \bfj \in I_{[n]}$.
If $\bfi(k)=\bfj(k)$, then the multiset $\{{\bf i}, \bar{\bf i}^{\{k\}^c}, \bar{\bf j}, \bar{\bf j}^{\{k\}}\}$
is balanced, so we have $a_{\bf i}a_{\bar{\bf i}^{\{k\}^c}}=a_{\bf j}a_{\bar{\bf j}^{\{k\}^c}}$ by (vi).
If $\bfi(k) \ne \bfj(k)$, then the multiset $\{{\bf i}, \bar{\bf i}^{\{k\}^c}, \bfj, \bar{\bf j}^{\{k\}^c}\}$
is balanced, so we also have $a_{\bf i}a_{\bar{\bf i}^{\{k\}^c}}
=1 \slash a_{\bf j}a_{\bar{\bf j}^{\{k\}^c}}$.
It remains to show that $r=(r_1,\dots,r_n)$ is a solution of (\ref{equa-sep}).
For a given index ${\bf i}=i_1 i_2 \cdots i_n$, we consider indices
$$
\bfi_k = \bar i_1 \bar i_2 \cdots \bar i_{k-1} i_k \cdots i_n, \qquad \bfj_k = i_1 i_2 \cdots i_k \bar i_{k+1} \cdots \bar i_n.
$$
Then, they satisfy
$\bfi_k(k)=i_k$, ${\bf j}_k=\bar{\bf i}_k^{\{k\}^c}$, $\bar{\bf j}_{k-1}={\bf i}_{k}$ and $\bfi_1=\bfi=\bfj_n$.
It follows that
$$
\begin{aligned}
(r^{\bf i})^2
&= (r_1^2)^{1-2i_1} (r_2^2)^{1-2i_2} \dots (r_n^2)^{1-2i_n} \\
&=(a_{{\bf i}_1}a_{{\bf j}_1})(a_{{\bf i}_2}a_{{\bf j}_2})\dots (a_{{\bf i}_n}a_{{\bf j}_n})\\
&=a_{{\bf i}_1}(a_{{\bf j}_1}a_{{\bf i}_2})(a_{{\bf j}_2}a_{{\bf i}_3})\dots (a_{{\bf j}_{n-1}}a_{{\bf i}_n})a_{{\bf j}_n}
=a_{{\bf i}_1}a_{{\bf j}_n}=(a_{\bf i})^2,
\end{aligned}
$$
as it was required. The equation $\alpha^{\bf i}=c_{\bf i}$ can be solved similarly.
\end{proof}

The argument above shows that the system (\ref{equa-sep}) of equations with unknowns $r_k$'s
has a unique solution
whenever the condition in (vi) is satisfied.
In the three qubit case, it is easy to see that there are only two irreducible balanced multisets of
order greater than or equal to $4$:
$\{000,011,101,110\}$ and $\{111,100,010,001\}$. Therefore, a three qubit {\sf X}-state
$\varrho=X(a,c)$ of rank four with
$a_{\bf i}a_{\bar{\bf i}}=|c_{\bf i}|^2=1$ ($\bfi\in I_{[3]}$) is separable if and only if
the identities
$$
a_{000}a_{011}a_{101}a_{110}=1,\qquad c_{000}c_{011}c_{101}c_{110}=1
$$
hold. This recovers a result in \cite{han_kye_phase}.
We also have
$$
\begin{aligned}
\tilde\Delta_3(a)
=\min\{\sqrt{a_{000}a_{111}},
\sqrt{a_{001}a_{110}},
&\sqrt{a_{010}a_{101}},
\sqrt{a_{011}a_{100}},\\
&\sqrt[4]{a_{000}a_{011}a_{101}a_{110}},
\sqrt[4]{a_{111}a_{100}a_{010}a_{001}}\}.
\end{aligned}
$$
This number  appears in the G\"uhne's separability criterion \cite{guhne_pla_2011}.
We show that the equality $\Delta_3(a)=\tilde\Delta_3(a)$ holds for the three qubit case,
from which we recover the main result in \cite{chen_han_kye}.
It would be nice to know if the identity $\Delta_n(a)=\tilde\Delta_n(a)$ holds for $n\ge 4$.

\begin{proposition}
We have the identity $\Delta_3(a)=\tilde\Delta_3(a)$ for every $a\in\calv_3^+$.
\end{proposition}

\begin{proof}
In order to prove $\tilde\Delta_3(a)\le \Delta_3(a)$, it suffices to show that $\tilde\Delta_3(a)=1$ implies
$\Delta_3(a)\ge 1$. Note that $\tilde\Delta_3(a)=1$ implies
$$
\max\{ a_{000}^{-1}a_{110}^{-1}a_{101}^{-1},\ a_{100}^{-1}\}
\le
\min\{ a_{111}a_{001}a_{010},\ a_{011}\}.
$$
Take $b_{011}$ between these two numbers. Because all the three intervals
$[a_{000}^{-1},a_{111}]$, $[a_{110}^{-1},a_{001}]$ and $[a_{101}^{-1},a_{010}]$ are nonempty, we can take
$b_{000}, b_{001},b_{010}$ so that
$$
\begin{aligned}
&b_{000}^{-1}b_{001}b_{010}=b_{011},\\
&a_{000}^{-1}\le b_{000}^{-1}\le a_{111},\quad
a_{110}^{-1}\le b_{001}\le a_{001},\quad
a_{101}^{-1}\le b_{010}\le a_{010}.
\end{aligned}
$$
Take $b\in\calv_3^+$ so that $b_\bfi b_{\bar\bfi}=1$ for each $\bfi\in I_{[3]}$. Then we see that
$X(a,\bf 1)$ is the sum of $X(b,\bf 1)$ and a diagonal state, with ${\bf 1}=(1,1,\dots, 1)$.
By Theorem \ref{sep-half-rank}, we see that $X(b,\bf 1)$ is separable, and so is $X(a,\bf 1)$. Since $\|{\bf 1}\|_{{\xx}_3}^\prime=1$,
we have $\Delta_3(a)\ge 1$ by Theorem \ref{x-sep-th}.
\end{proof}

Now, we look for balanced multisets. To do this, we first note that
the following are equivalent:
\begin{itemize}
\item
The multiset $\{\bfi_1, \cdots \bfi_m, \bar \bfj_1, \cdots \bar \bfj_m \}$ is balanced,
\item
its {\sl conjugate} $\{\bar \bfi_1, \cdots \bar \bfi_m, \bfj_1, \cdots \bfj_m \}$ is balanced,
\item
$\bfi_1(k) + \bfi_2(k) + \cdots + \bfi_m(k) = \bfj_1(k) + \bfj_2(k) + \cdots + \bfj_m(k)$ for every $k=1,2,\dots,n$.
\end{itemize}
We use the notation
\begin{equation}\label{notation_balanced}
\bfi_1 + \bfi_2 + \cdots + \bfi_m \equiv \bfj_1 + \bfj_2 + \cdots + \bfj_m
\end{equation}
whenever the last condition holds.
If we identify an index ${\bf i}=i_1i_2\dots i_n$ with the natural number
${\bf i}=\sum_{k=1}^n i_k2^{n-k}$ using binary expansion, then we see that the relation (\ref{notation_balanced})
implies the identity $\bfi_1 + \bfi_2 + \cdots + \bfi_m = \bfj_1 + \bfj_2 + \cdots + \bfj_m$ as natural numbers.
Note that  the converse does not hold. In the three qubit case,
the relation $000+011 \equiv 001+010$, or equivalently $0+3\equiv 1+2$
represents two balanced multisets $\{000,011,110,101\}$ and $\{111,100,001,010\}$. It is easily checked that they are irreducible,
and there is no more irreducible balanced multiset of order four. Note that we may assume that all the indices $\bfi_k$ and $\bfj_k$
begin with $0$, when we look for multisets of the form
$\{\bfi_1, \cdots \bfi_m, \bar \bfj_1, \cdots \bar \bfj_m \}$.

Irreducible balanced multisets of order four in the four qubit system can be expressed by the identities
\begin{equation}\label{4qubit4order}
\begin{aligned}
0+7 \equiv 1+6 \equiv 2+5 \equiv 3+4, \\
0+3 \equiv 1+2,\quad 4+7 \equiv 5+6,\\
0+5 \equiv 1+4,\quad 2+7 \equiv 3+6,\\
0+6 \equiv 2+4,\quad 1+7 \equiv 3+5.
\end{aligned}
\end{equation}
One may check that these are all possible identities, and so we have $(6+6)\times 2=24$ irreducible balanced multisets
of order four. By a simple combinatorial method, one may also check that there are $8\times 2=16$ irreducible balanced multisets of order six,
which can be
expressed by the identities
$$
\begin{aligned}
0+0+7 \equiv 1+2+4, \qquad& 1+1+6 \equiv 0+3+5, \\
2+2+5 \equiv 0+3+6, \qquad& 3+3+4 \equiv 1+2+7, \\
3+4+4 \equiv 0+5+6, \qquad& 2+5+5 \equiv 1+4+7, \\
1+6+6 \equiv 2+4+7, \qquad& 0+7+7 \equiv 3+5+6.
\end{aligned}
$$
The relation
$0+0+7 \equiv 1+2+4$, or equivalently $0000+0000+0111 \equiv 0001+0010+0100$
represents two irreducible balanced multisets
$$
\{0000,0000,0111,1110,1101,1011\}, \qquad
\{1111,1111,1000,0001,0010,0100\}.
$$
It is not difficult to show that there is no irreducible balanced multiset of order eight.
Consequently, we found all the irreducible balanced multisets in the four qubit case, which gives rise to
the number $\tilde\Delta_4(a)$ for $a\in\calv_4^+$. This number is given by the minimum of $8+24+16=48$ numbers.

In general, we can find all irreducible balanced multisets of order four in the $n$-qubit system, in a recursive way.
For an index ${\bf i}=i_1i_2\ldots i_{n+1} \in I_{[n+1]}$, we define the index $\tilde {\bf i}=i_2\ldots i_{n+1} \in I_{[n]}$
by deleting the leftmost bit of ${\bf i}$, and denote $\mathcal T_{n,m}$ be the family of all irreducible balanced multisets
of order $m$ in the $n$-qubit system.
If $T=\{{\bf i}_1, {\bf i}_2, {\bf i}_3, {\bf i}_4\}$ is in $\mathcal T_{n+1, 4}$,
then the balanced multiset $\{\tilde{\bf i}_1, \tilde{\bf i}_2, \tilde{\bf i}_3, \tilde{\bf i}_4\}$  in $n$-qubit system is
one of the following:
\begin{itemize}
\item
an irreducible balanced multiset of order four,
\item
a disjoint union of two irreducible balanced multisets of order 2.
\end{itemize}

Now, we consider the reverse direction.
For a given multiset $\{{\bf i}_1, {\bf i}_2, {\bf i}_3, {\bf i}_4 \}$ in $\mathcal T_{n,4}$,
we can construct the following six multisets in $\mathcal T_{n+1,4}$:
\[
\begin{aligned}
\{0{\bf i}_1, 0{\bf i}_2, 1{\bf i}_3, 1{\bf i}_4\},\quad
\{0{\bf i}_1, 1{\bf i}_2, 0{\bf i}_3, 1{\bf i}_4\},\quad
\{0{\bf i}_1, 1{\bf i}_2, 1{\bf i}_3, 0{\bf i}_4\},\\
\{1{\bf i}_1, 1{\bf i}_2, 0{\bf i}_3, 0{\bf i}_4\},\quad
\{1{\bf i}_1, 0{\bf i}_2, 1{\bf i}_3, 0{\bf i}_4\},\quad
\{1{\bf i}_1, 0{\bf i}_2, 0{\bf i}_3, 1{\bf i}_4\}.
\end{aligned}
\]
For any given two disjoint multisets $\{{\bf i},\bar{\bf i}\}$ and $\{{\bf j},\bar{\bf j}\}$ in $\mathcal T_{n,2}$,
we also obtain two multisets
\[
\{0{\bf i}, 0\bar{\bf i}, 1{\bf j}, 1\bar{\bf j}\},\quad \{1{\bf i}, 1\bar{\bf i}, 0{\bf j}, 0\bar{\bf j}\}
\]
in $\mathcal T_{n+1,4}$. So, we can construct $2 \cdot {{2^{n-1}}\choose{2}}=2^{n-1}(2^{n-1}-1)$
multisets in ${\calt_{n+1,4}}$ from $2^{n-1}$ multisets in $\calt_{n,2}$.
Consequently, $\mathcal T_{n+1,4}$ can be obtained inductively from $\mathcal T_{n,4}$ and $\mathcal T_{n,2}$,
with the following recursion formula:
\[
\# \mathcal T_{3,4}=2,\qquad \# \mathcal T_{n+1,4} = 6 (\#\mathcal T_{n,4})+2^{n-1}(2^{n-1}-1),\quad n=3,4,\dots.
\]
For example, we have $\#\calt_{4,4}=24$, $\#\calt_{5,4}=200$ and $\#\calt_{6,4}=1440$.

\section{Phase identities}\label{sec-phase}

We take logarithm on the system (\ref{equa-sep}) of equations
with $n$ unknowns and $2^n$ equations, and write $R_k=\log r_k$ for
$k=1,2,\dots, n$, and $A_{\bf i}=\log a_{\bf i}$ for ${\bf i}\in
I_n$. Then we have the equation $\Theta_nR=A$ with the linear map $\Theta_n:\mathbb R^n\to\calvnR$
defined in (\ref{theta}). Therefore, the equation (\ref{equa-sep}) has a solution if and only if
$A=\log a\in\image\Theta_n$. This is also the case for the phase part of anti-diagonal $c\in\calv_n^\sa$:
The equation $\tilde\alpha=c$ in Theorem \ref{sep-half-rank} (ii) has a solution if and only if the phase part
$\theta\in\calv_n^\ph$ of $c\in\calv_n^\sa$ belongs to the range of $\Theta_n$.

We take the orthonormal basis $\{e_\bfi:\bfi\in I_{[n]}\}$ of $\calv_n^\mathbb R$, then
the associated matrices with $\Theta_n$ with respect to these bases are given by
$$
\Theta_1=
\left(\begin{matrix}
+\\ -\end{matrix}\right),\qquad
\Theta_2=
\left(\begin{matrix}
+&+\\
+&-\\
-&+\\
-&-
\end{matrix}\right),\qquad
\Theta_3=
\left(\begin{matrix}
+&+&+\\
+&+&-\\
+&-&+\\
+&-&-\\
-&+&+\\
-&+&-\\
-&-&+\\
-&-&-
\end{matrix}\right),
$$
for $n=1,2,3$, where $+$ and $-$ represent $+1$ and $-1$, respectively.
We recall that a vector $\theta\in\calv_n^\ph$ satisfies the phase identities when $\theta\in\image\Theta_n$.
We rephrase the condition (ii) in Theorem \ref{sep-half-rank} to get the following:

\begin{proposition}\label{half-phase}
Let $\varrho=X(a,c)$ be an $n$-qubit {\sf X}-state with $a_{\bf i}a_{\bar{\bf i}}= |c_{\bf i}|^2=1$ for every index ${\bf i}$.
Then $\varrho$ is separable if and only if both $\log a\in\calv_n^\ph$ and the phase part $\theta\in\calv_n^\ph$ of $c$
satisfy the phase identities.
\end{proposition}

In order to check if $\theta$ satisfies the phase identities, we have to consider the orthogonal complement
$\calv_n^\ph\ominus \image\Theta_n$, which is a real vector space of dimension $\lambda(n):=2^{n-1}-n$.
For a given multiset $T$ of $n$ indices, we define the vector
$$
\xi_T=\frac 12\sum_{\bfi\in T}(e_\bfi - e_{\bar\bfi})\in \calv_n^\ph.
$$
Then, the multiset $T$ is balanced if and only if the vector $\xi_T$ belongs to $\calv_n^\ph\ominus \image\Theta_n$  because
we have the identity:
$$
\begin{aligned}
\left\langle \xi_T, \Theta(e_k) \right\rangle
&=\left\langle \frac 12\sum_{\bfi\in T}(e_\bfi - e_{\bar\bfi}), \Theta(e_k) \right\rangle\\
&=  \frac 12\left(\sum_{\bfi\in T} \Theta(e_k)_\bfi - \Theta(e_k)_{\bar \bfi}\right) \\
&= \# \{ \bfi \in T : \bfi (k)=0 \} - \# \{ \bfi \in T : \bfi (k)=1 \}.
\end{aligned}
$$
We say that a family $\calt$ of irreducible balanced multisets is {\sl basic} when $\{\xi_T:T\in\calt\}$
is a basis of $\calv_n^\ph\ominus \image\Theta_n$. We note that $\lan \theta,\xi_T\ran=\sum_{\bfi\in T}\theta_\bfi$,
which implies the following:

\begin{proposition}\label{phid}
Let $\calt$ be a basic family of irreducible balanced multisets. For $\theta\in\calv_n^\ph$, the following are equivalent:
\begin{enumerate}
\item[(i)]
$\theta$ satisfies the phase identities;
\item[(ii)]
$\sum_{\bfi\in T}\theta_\bfi=0$ for every $T\in\calt$;
\item[(iii)]
$\sum_{\bfi\in T}\theta_\bfi=0$ for every balanced multiset $T$.
\end{enumerate}
\end{proposition}

In the three qubit case, any basic family consists of a single multiset $\{000, 011, 101, 110\}$
or its conjugate $\{111,100, 010, 001\}$.
Therefore, $\theta\in\calv_3^\ph$ satisfies the phase identities if and only if
the identity $\theta_{000}+\theta_{011}+\theta_{101}+\theta_{110}=0$, or equivalently
$$
\theta_{000}+\theta_{011}=\theta_{001}+\theta_{010}
$$
holds, which arises from the relation
$000+011\equiv 001+010$.
This is exactly the phase identity
considered in \cite{han_kye_phase}.

The condition (vi) of Theorem \ref{sep-half-rank} suggests that there exists a basic family consisting of
irreducible balanced multisets of order four. To get a basic family of such kind, we take
an index which begins with $0$ and in which $1$ appears at least two
times. We call that a {\sl non-elementary} index.  They
correspond natural numbers which are not of the form $2^\ell$. We
decompose such number ${\bf i}$ as the sum of the smallest number of
the form $2^\ell$ with the same position of $1$ as ${\bf i}$ and the
other. In the three qubit case, we have only one non-elementary index
with the decomposition as follows:
\begin{equation}\label{pbms_3}
0+3=000+011 \equiv 001+010=1+2
\end{equation}
In the four qubit case, we have four non-elementary indices: One of them appears in the three qubit case, and the other three
cases are listed as follows:
\begin{equation}\label{pbms_4}
\begin{aligned}
0+5=0000+0101&\equiv 0001+0100=1+4\\
0+6=0000+0110&\equiv 0010+0100=2+4\\
0+7=0000+0111&\equiv 0001+0110=1+6\\
\end{aligned}
\end{equation}
There are $\lambda(5)=11$ non-elementary indices for the five qubit case. Four of them appear in the three and four qubit cases, and we list up
the other seven cases:
$$
\begin{aligned}
0+\phantom{1}9&\equiv 1+\phantom{1}8,\
0+10&\equiv 2+\phantom{1}8,\
0+11&\equiv 1+10,\
0+12\equiv 4+8,\\
0+13&\equiv 1+12,\
0+14&\equiv 2+12,\
0+15&\equiv 1+14.
\end{aligned}
$$

Formally, for a non-elementary index ${\bf i}=i_1i_2\dots i_n$, we take
the biggest $k$ such that $i_k=1$. Define ${\bf i}_{\min}\in I_n$ by
${\bf i}_{\min}(k)=1$ and ${\bf i}_{\min}(\ell)=0$ for $\ell\neq k$.
Put ${\bf i}_{\res}={\bf i}-{\bf i}_{\min}$.
Since ${\bf 0} + \bfi \equiv \bfi_{\min} + \bfi_{\res}$, we see that
\begin{equation}\label{non-ele}
T_{\bf i}=\{{\bf 0}, {\bf i}, \bar{\bf i}_{\min}, \bar{\bf i}_{\res}\}
\end{equation}
is an irreducible balanced multiset of order four, where ${\bf 0}=00\dots 0$.
We note that the number of non-elementary indices is exactly $2^{n-1}-n$,
which coincides with the dimension of $\calv_n^\ph\ominus \image\Theta_n$.
One may also verify that $\xi_{T_{\bfi}}$ with non-elementary indices $\bfi$'s are linearly independent,
so $\{T_{\bfi}\}$ with non-elementary indices $\bfi$'s is basic.
Therefore. $\theta\in\calv_4^\ph$ satisfies the phase identity if and only if the following four identities
\begin{equation}\label{4phaseidenity}
\begin{aligned}
\theta_{0000}+\theta_{0011}&=\theta_{0001}+\theta_{0010}\\
\theta_{0000}+\theta_{0101}&=\theta_{0001}+\theta_{0100}\\
\theta_{0000}+\theta_{0110}&=\theta_{0010}+\theta_{0100}\\
\theta_{0000}+\theta_{0111}&=\theta_{0001}+\theta_{0110}
\end{aligned}
\end{equation}
hold. These arise from the relations in (\ref{pbms_3}) and (\ref{pbms_4}).

Proposition \ref{half-phase} tells us that we need all the possible phase identities in Proposition \ref{phid}
to characterize the separability of half rank {\sf X}-states.
In some circumstances, parts of the phase identities
may be necessary for separability.

\begin{theorem}\label{corank-2}
Let $\varrho=X(a,c)$ be an $n$-qubit state with the phase part $\theta\in\calv_n^\ph$.
Suppose that there exist two indices ${\bf
i}_1,{\bf i}_2$ with ${\bf i}_1\neq{\bf i}_2$ and ${\bf
i}_1\neq\bar{\bf i}_2$ satisfying the relation
$$
\sqrt{a_{\bf i}a_{\bar{\bf i}}}=|c_{\bf i}|=\|c\|_\infty,\qquad {\bf i}={\bf i}_1,{\bf i}_2.
$$
If $\varrho$ is separable then we have the identities
\begin{equation}\label{yfjylgkh}
|c_{{\bf j}_1}|=|c_{{\bf j}_2}|,\qquad \theta_{{\bf i}_1}+ \theta_{{\bf i}_2}=\theta_{{\bf j}_1}+\theta_{{\bf j}_2}\mod 2\pi,
\end{equation}
whenever $\{{\bf i}_1,{\bf i}_2,\bar{\bf j}_1,\bar{\bf j}_2\}$ is an irreducible balanced multiset.
\end{theorem}

\begin{proof}
We write $\varrho = \sum_k \lambda_k \omega_k$ with $\lambda_k>0$, $\sum_k \lambda_k=1$
and pure product states $\omega_k$, whose {\sf X}-parts
are given by $X(a^k,c^k)$ with $a^k \in \calv_n^+$ and $c^k \in \calv_n^\sa$.
By the Cauchy-Schwartz inequality, we have
$$
\begin{aligned}
|c_\bfi| = |\sum_k \lambda_k c^k_\bfi| & \le \sum_k |\lambda_k c^k_\bfi| \\
& \le \sum_k \lambda_k \sqrt{a^k_\bfi a^k_{\bar \bfi}} \\
& \le (\sum_k \lambda_k a^k_\bfi)^{1 \slash 2} (\sum_k \lambda_k a^k_{\bar \bfi})^{1 \slash 2} = \sqrt{a_\bfi a_{\bar \bfi}},
\end{aligned}
$$
which become identities for $\bfi=\bfi_1,\bfi_2$. By the first identity
$|\sum_k \lambda_k c^k_\bfi| = \sum_k |\lambda_k c^k_\bfi|$ with $\bfi= \bfi_1, \bfi_2$,
we have $\theta_\bfi=\arg c^k_\bfi$ for $\bfi=\bfi_1,\bfi_2$.
Since the {\sf X}-part of a pure state is diagonal or non-diagonal of rank $2^{n-1}$, we have
$\sqrt{a^k_\bfi a^k_{\bar \bfi}}=|c^k_\bfj|$ for every $\bfi,\bfj\in I_{[n]}$, whenever $c^k \ne 0$ by the PPT condition.
This number will be denoted by $r_k$. It follows that
$$
c_\bfi^k = r_k e^{{\rm i} \theta_\bfi} \quad (\bfi=\bfi_1,\bfi_2) \qquad \text{and} \qquad \|c\|_\infty
=|c_{\bfi_1}|=|c_{\bfi_2}|=\sum_k \lambda_k r_k.
$$
Therefore, we have
\begin{equation}\label{jhgiugkjl}
|c_{\bfj_1}|^2 = c_{\bfj_1} \bar c_{\bfj_1}  = \left(\sum_k \lambda_k r_k e^{{\rm i}\theta_{\bfj_1}^k}\right)~
\left(\sum_l \lambda_l r_l e^{-{\rm i}\theta_{\bfj_1}^\ell}\right)
=\sum_{k,\ell}\lambda_k \lambda_\ell r_k r_\ell e^{{\rm i} (\theta_{\bfj_1}^k-\theta_{\bfj_1}^\ell)}
\end{equation}
with $\theta_\bfi^k=\arg c^k_\bfi$ for $c^k \ne 0$.
Now, suppose that  $\{{\bf i}_1,{\bf i}_2, \bar{\bf j}_1, \bar{\bf j}_2\}$ is an irreducible balanced multiset.
Then we have
\begin{equation}\label{phase-idxxx}
\theta_{\bfj_1}^k + \theta_{\bfj_2}^k = \theta_{\bfi_1}^k + \theta_{\bfi_2}^k = \theta_{\bfi_1} + \theta_{\bfi_2}
= \theta_{\bfi_1}^\ell + \theta_{\bfi_2}^\ell = \theta_{\bfj_1}^\ell + \theta_{\bfj_2}^\ell, \mod 2\pi
\end{equation}
by Theorem \ref{sep-half-rank}.
Therefore, $\theta_{\bfj_1}^k-\theta_{\bfj_1}^\ell$ may be replaced by $\theta_{\bfj_2}^\ell-\theta_{\bfj_2}^k$ in (\ref{jhgiugkjl}),
and so we have the identity
$|c_{{\bf j}_1}|=|c_{{\bf j}_2}|$.

We also have
$$
\begin{aligned}
c_{\bfj_1} c_{\bfj_2}
&= \left(\sum_k \lambda_k r_k e^{{\rm i}\theta_{\bfj_1}^k}\right)~ \left(\sum_\ell \lambda_\ell r_\ell e^{{\rm i}\theta_{\bfj_2}^\ell}\right)\\
&= \sum_{k,\ell} \lambda_k \lambda_\ell r_kr_\ell~ e^{{\rm i}(\theta_{\bfj_1}^k+\theta_{\bfj_2}^\ell)}\\
&=  {1 \over 2}\sum_{k,\ell} \lambda_k \lambda_\ell r_k r_\ell~ e^{{\rm i}(\theta_{\bfj_1}^k+\theta_{\bfj_2}^\ell)}
     +\lambda_\ell \lambda_k r_\ell r_k~ e^{{\rm i}(\theta_{\bfj_1}^\ell+\theta_{\bfj_2}^k)}\\
&= {1 \over 2}\sum_{k,\ell} \lambda_k \lambda_\ell r_k r_\ell
   \left(e^{{\rm i}(\theta_{\bfj_1}^k+\theta_{\bfj_2}^k)} e^{{\rm i}(\theta_{\bfj_2}^l-\theta_{\bfj_2}^k)}
   +e^{{\rm i}(\theta_{\bfj_1}^\ell+\theta_{\bfj_2}^\ell)} e^{{\rm i}(\theta_{\bfj_2}^k-\theta_{\bfj_2}^\ell)}\right).
\end{aligned}
$$
Applying the identity (\ref{phase-idxxx}), we have
$$
c_{\bfj_1} c_{\bfj_2}
=e^{{\rm i}(\theta_{\bfi_1}+\theta_{\bfi_2})}\sum_{k,\ell} \lambda_k \lambda_\ell r_k r_\ell
    \cos (\theta_{\bfj_2}^\ell-\theta_{\bfj_2}^k).
$$
We see that the phase part of $c_{\bfj_1} c_{\bfj_2}$ is given by $\theta_{\bfi_1}+\theta_{\bfi_2}$, because
$$
\begin{aligned}
\sum_{k,l} \lambda_k \lambda_l r_k r_l \cos (\theta_{\bfj_2}^l-\theta_{\bfj_2}^k)
& = \sum_{k,l} \lambda_k \lambda_l r_k r_l (\cos \theta_{\bfj_2}^l \cos \theta_{\bfj_2}^k+\sin \theta_{\bfj_2}^l \sin \theta_{\bfj_2}^k)\\
& = \left(\sum_k \lambda_k r_k \cos \theta_{\bfj_2}^k\right)^2
    +\left(\sum_k \lambda_k r_k \sin \theta_{\bfj_2}^k\right)^2 \ge 0.
\end{aligned}
$$
This implies the required identity
$\theta_{{\bf i}_1}+ \theta_{{\bf i}_2}=\theta_{{\bf j}_1}+\theta_{{\bf j}_2}$.
\end{proof}

Note that Theorem \ref{corank-2} may be applied whenever a separable {\sf X}-state $\varrho$ has corank $\ge 2$.
In the three qubit case, we apply Theorem \ref{corank-2} for $\{\bfi_1,\bar\bfi_2\}=\{000,110\}$. In this case, we see that
$\{{\bf i}_1,\bar{\bf i}_2,\bar{\bf j}_1, \bfj_2\}$ is an irreducible balanced multiset if and only if
$\{\bar \bfj_1,\bfj_2\}=\{101,011\}$,
which implies the separability criteria $|c_{010}|=|c_{100}|$ and $\theta_{000}+ \theta_{110}=\theta_{010}+\theta_{100}$,
or equivalently $|c_{010}|=|c_{011}|$ and $\theta_{000}+\theta_{011}=\theta_{001}+\theta_{010}$.
In this way, we see that if a three qubit {\sf X}-state $\varrho=X(a,c)$ of rank six is separable then
there exists a partition $\{\bfi_1,\bfi_2\}\cup\{\bfj_1,\bfj_2\}$ of indices beginning $0$ such that
$$
|c_{{\bf i}_1}|=|c_{{\bf i}_2}|,\qquad
|c_{{\bf j}_1}|=|c_{{\bf j}_2}|,\qquad \theta_{000}+ \theta_{011}=\theta_{001}+\theta_{010}.
$$
Note that the {phase identity} $\theta_{000}+ \theta_{011}=\theta_{001}+\theta_{010}$
must be retained regardless of partition.
Theorem 5.1 of \cite{han_kye_phase} tells us that this is essentially a characterization of
separability of a three qubit {\sf X}-state of rank six.

\section{phase differences}\label{sec-phase-diff}

The phases also play important roles to investigate properties of the numbers $\|u\|_{\xx_n}$ and $\|c\|^\prime_{\xx_n}$
in the characterization of block-positivity and separability.
Suppose that $u\in \calv_n^\sa$ has the phase part $\phi\in\calv_n^\ph$.
We write $\alpha_k=e^{{\rm i}\theta_k}$ for each $k=1,2,\dots, n$. Then we have
$$
\alpha^\bfi =\prod_{k=1}^n e^{{\rm i}(1-2\bfi(k))\theta_k} = e^{{\rm i} \sum_{k=1}^n(1-2\bfi(k))\theta_k}
=e^{{\rm i}\Theta_n(\theta)_\bfi}
$$
for each $\bfi\in I_{[n]}$, and so it follows that
$$
\sum_{\bfi\in I_{[n]}} u_\bfi \alpha^\bfi
=\sum_{\bfi\in I_{[n]}}|u_\bfi|e^{{\rm i}(\phi_\bfi+\Theta_n(\theta)_\bfi)}
=\langle |u|, e^{{\rm i}(\phi+\Theta_n(\theta))}\rangle.
$$
Therefore, we see that the range of the function
$\alpha  \mapsto \sum_{\bfi\in I_n} u_\bfi \alpha^\bfi \in \mathbb R$ defined on $\mathbb T^n$
is determined the modulus vector $|u|$ and the coset $\phi+\image\Theta_n$ in the quotient space $\calv_n^\ph/\image\Theta_n$.
The {\sl phase difference} of $u\in\calv_n^\sa$ with nonzero entries
is defined by the coset to which its phase part $\phi$ belongs. It is also called the phase difference of $\phi\in\calv_n^\ph$.
The phase differences of $\phi,\psi\in\calv_n^\ph$ coincide if and only if
$\phi-\psi$ satisfies the phase identities. The phase difference is uniquely expressed by a vector in
the space $\calv_n^\ph\ominus\image\Theta_n$, which will be denoted by $\Phi_n(u)$.
We see that $u\in\calv_n^\sa$ satisfies the phase identities if and only if
it has the trivial phase difference. We summarize our discussion as follows:

\begin{proposition}\label{norm-phase}
If $u,v\in \calv_n^\sa $ satisfy $\Phi_n(u)=\Phi_n(v)$ and $|u_\bfi|=|v_\bfi|$ for each $\bfi\in I_{[n]}$, then we have
$\|u\|_{\xx_n}=\|v\|_{\xx_n}$.
\end{proposition}

In the three qubit case, we note that $\calv_3^\ph\ominus\image\Theta_3$ is
of dimension $\lambda(3)=1$ and it is spanned by the vector $\xi_T$ with
the balanced multiset $T=\{000, 011, 101, 110\}$. In this case, the phase difference
can be expressed by the scalar
$$
\theta_{000}-\theta_{001}-\theta_{010}+\theta_{011}=
\lan\Phi_n(c),\xi_T\ran.
$$
This is exactly the phase difference of a three qubit state introduced in \cite{han_kye_phase}.
The following theorem tells us that the separability of an {\sf X}-shaped multi-qubit states
depends only on the phase differences of the anti-diagonal parts, as well as magnitudes of diagonal and
anti-diagonal parts.

\begin{theorem}\label{dual-norm-phase}
Suppose that $c,d\in \calv_n^\sa$
satisfy $\Phi_n(c)=\Phi_n(d)$ and
$|c_\bfi|=|d_\bfi|$ for each $\bfi\in I_{[n]}$. Then we have $\|c\|_{\xx_n}^\prime=\|d\|_{\xx_n}^\prime$.
\end{theorem}

\begin{proof}
For  a given $z\in \calv_n^\sa $ with $\|z\|_{\xx_n}=1$, we define $w\in \calv_n^\sa $ by
$$
w_\bfi=|z_\bfi|e^{{\rm i}(\arg z_\bfi +\arg c_\bfi-\arg d_\bfi)}, \qquad \bfi\in I_{[n]}.
$$
Since $\arg c-\arg d\in \image\Theta_n$ by $\Phi_n(c)=\Phi_n(d)$, we have
$$
\arg w+\image\Theta_n=\arg z+\arg c-\arg d+\image\Theta_n =\arg z+\image\Theta_n,
$$
and so we have $\|z\|_{\xx_n}=\|w\|_{\xx_n}$ by Proposition \ref{norm-phase}. We also have
$$
\langle c,z\rangle
=\sum_{\bfi\in I_{[n]}}|c_\bfi||z_\bfi|e^{{\rm i}(\arg c_\bfi+\arg z_\bfi)}
=\sum_{\bfi\in I_{[n]}}|d_\bfi||w_\bfi|e^{{\rm i}(\arg d_\bfi+\arg w_\bfi)}
=\langle d,w\rangle.
$$
We have shown that for every $z\in \calv_n^\sa $ with $\|z\|_{\xx_n}=1$ there exists $w\in \calv_n^\sa $ such that
$\|w\|_{\xx_n}=1$ and $\langle c,z\rangle=\langle d,w\rangle$. This
implies that $\|c\|_{\xx_n}^\prime\le \|d\|_{\xx_n}^\prime$. The reverse inequality holds
by the same argument.
\end{proof}

Recall the inequality $\|c\|_{\xx_n}^\prime\ge \|c\|_\infty$ in (\ref{dual-norm-ineq}).
Now, we proceed to get another lower bound of the dual norm
$\|c\|_{\xx_n}^\prime$ for $c\in \calv_n^\sa $ in terms of the norms $\|\cdot\|_{\xx_n}$ of some elements in $\calv_n^\sa $.
To do this, we take $\phi\in \calv_n^\ph$ to get
$$
\|c\|_{\xx_n}^\prime
\ge \frac{\langle c,\tilde\alpha\circ e^{{\rm i}\phi}\rangle}{\|\tilde\alpha\circ e^{{\rm i}\phi}\|_{\xx_n}}
$$
for every $\alpha\in\mathbb T^n$, where $a\circ b$ is define by $(a\circ b)_\bfi=a_\bfi b_\bfi$.
Now, we have
$$
\langle c,\tilde\alpha\circ e^{{\rm i}\phi}\rangle=\langle c\circ e^{{\rm i}\phi}, \tilde\alpha\rangle,\qquad
\|\tilde\alpha\circ e^{{\rm i}\phi}\|_{\xx_n}=\|e^{{\rm i}\phi}\|_{\xx_n}
$$
by (\ref{A_nB_n}) or Proposition \ref{norm-phase}, because $\Phi_n(\tilde\alpha)={\bf 0}$. Therefore, we have the inequality
\begin{equation}\label{ineq-esi-dual}
\|c\|_{\xx_n}^\prime
\ge \max_{\phi\in\calv_n^\ph}\frac{\|c\circ e^{{\rm i}\phi}\|_{\xx_n}}{\|e^{{\rm i}\phi}\|_{\xx_n}},
\end{equation}
for every $c\in\calv_n^\sa$. In the three qubit case of $n=3$, this is exactly Proposition 5.6 of \cite{chen_han_kye}.
The number $\|u\|_{\xx_n}$ has a natural upper bounds $\|u\|_1$ by Proposition \ref{b_n_ineq}.
If all the entries of $u$ is nonnegative then it is clear that $\|u\|_{\xx_n}=\|u\|_1$
by taking $\alpha={\bf 1}$ in (\ref{A_nB_n}).
We investigate when the strict inequalities hold in $\|u\|_{\xx_n}\le\|u\|_1$ and $\|c\|_{\xx_n}^\prime\ge \|c\|_\infty$.

\begin{theorem}\label{sep-ppt}
For $u,c\in\calv_n^\sa$, we have the following:
\begin{enumerate}
\item[(i)]
$\|u\|_{\xx_n}= \|u\|_1$ holds if and only if $u$ has the trivial phase difference;
\item[(ii)]
Suppose that $|c_\bfi|=|c_{\bfj}|$ for every $\bfi,\bfj\in I_{[n]}$. Then $\|c\|^\prime_{\xx_n}= \|c\|_\infty$ if and only if $c$ has
the trivial phase difference.
\end{enumerate}
\end{theorem}

\begin{proof}
If the phase part of $u$ is ${\bf 0}\in\calv_n^\ph$, then we have $\|u\|_{\xx_n}=\|u\|_1$.
This is also the case whenever $\Phi_n(u)=0$ by Proposition \ref{norm-phase}.
We note that
$$
u_\bfi\alpha^\bfi+u_{\bar\bfi}\alpha^{\bar\bfi}
\le |u_\bfi\alpha^\bfi+u_{\bar\bfi}\alpha^{\bar\bfi}|\le |u_\bfi|+|u_{\bar\bfi}|
$$
for each $\bfi\in I_{[n]}$.
If $\|u\|_{\xx_n}=\|u\|_1$ then there exists $\alpha\in\mathbb T^n$ such that
 $u_{\bf i}\alpha^{\bf i}=|u_\bfi|$ for every $\bfi\in I_{[n]}$ in
(\ref{bbbbbb}). Therefore, the phase part of $u$ is given by $\theta$ with $\theta_\bfi=-\arg\alpha^\bfi$,
which belongs to $\image\Theta_n$. This proves the statement (i).

To prove (ii), we may assume that $|c_\bfi|=1$ for every $\bfi\in I_{[n]}$. If $\|u\|_{\xx_n}=1$ then we have
$\lan u,{\bf 1}\ran =\sum_{\bfi\in I_{[n]}}u_{\bfi}\le \|u\|_{\xx_n}=1$,
which implies that $\|{\bf 1}\|^\prime_{\xx_n}\le 1$. Therefore, we have $\|{\bf 1}\|^\prime_{\xx_n}= 1$ by (\ref{dual-norm-ineq}),
and $\|c\|^\prime_{\xx_n}= 1$ whenever $c$ has a trivial phase difference by Theorem \ref{dual-norm-phase}.
For the converse, suppose that $\|c\|^\prime_{\xx_n}= 1$. Then we have
$\|e^{{\rm i}\phi}\|_{\xx_n} \ge \|c\circ e^{{\rm i}\phi}\|_{\xx_n}$ for every $\phi\in\calv_n^\ph$ by (\ref{ineq-esi-dual}).
Take $\phi=-\theta$ with the phase part $\theta$ of $c\in\calv_n^\sa$. Then we have
$$
\|e^{{\rm i}(-\theta)}\|_{\xx_n} \ge \|c\circ e^{{\rm i}(-\theta)}\|_{\xx_n}= \|{\bf 1}\|_{\xx_n}=\|{\bf 1}\|_1=\|e^{{\rm i}(-\theta)}\|_1.
$$
This implies that $e^{{\rm i}(-\theta)}$ has the trivial phase difference by (i), and completes the proof.
\end{proof}

Finally, we show that our criterion, Theorem \ref{criterion}, detects
nonzero volume of PPT entanglement with respect to the affine space of all self-adjoint
matrices with trace one. To do this, we consider the
general situations to compare two convex sets. Let $C$ be a convex set in a
finite dimensional real vector space. Recall that a point $p\in C$ is called an
interior point when it is a topological interior point of $C$ with
respect to the affine manifold generated by $C$, and a boundary
point if it is not an interior point. We denote by $\inte C$ and
$\partial C$ the sets of all interior points and the boundary points of $C$,
respectively. If two convex sets $C_1$ and $C_2$ generate the common affine manifold $M$, then we may use
all the topological notions, like interior, closure, boundary, with respect $M$.

\begin{proposition}\label{general}
Let $C_1\subset C_2$ be a finite dimensional convex sets which generate a common affine manifold.
Then the following are equivalent:
\begin{enumerate}
\item[(i)]
$C_2\setminus C_1$ has the nonempty interior;
\item[(ii)]
$\partial \overline{C_1}\cap\inte C_2$ is nonempty.
\end{enumerate}
\end{proposition}

\begin{proof}
The assumption implies that $\inte C_1\subset \inte C_2$.
We fix a common interior point $p_0$ of $C_1$ and $C_2$. For the direction (i) $\Longrightarrow$ (ii),
take an interior point $p_1$ of $C_2\setminus C_1$, and consider the line segment $p_t=(1-t)p_0+tp_1$.
Put $t_0=\sup\{t:p_t\in\overline{C_1}\}$. Then $t_0>0$ since $p_0$ is an interior point of $C_1$.
We also have $t_0<1$, because $p_1$ is an interior point of $C_2\setminus C_1$. Therefore, we see that $p_{t_0}$
belongs to $\partial \overline{C_1}\cap\inte C_2$.

For (ii) $\Longrightarrow$ (i), take $p_1\in \partial \overline{C_1}\cap\inte C_2$ and
put $t_0=\sup\{t:p_t\in C_2\}$. Since $p_1$ is an interior point of $C_2$, we have $t_0>1$.
Then $p_t$ with $1<t<t_0$ is an interior point of $C_2$. Because $\overline{C_1}$ is convex, we conclude that
$p_t$ is an interior point of $C_2\setminus C_1$.
\end{proof}

Proposition \ref{general} has been used implicitly in \cite{{kye_rev},{kye_multi_dual}},
to see that an entanglement witness $W$ detects nonzero volume of PPT entanglement if and only if
every partial transpose of $W$ has the full rank.
We denote by ${\mathcal C}$ the set of all states satisfying the criterion in Theorem \ref{criterion}.
We show that ${\mathcal C}$ is convex and closed.

\begin{proposition}\label{convex_general}
The set $\mathcal C$ is convex and closed.
\end{proposition}

\begin{proof}
For given $a,b\in \calv_n^+$ and $c,d\in \calv_n^\sa$, we have the inequalities
$$
\Delta_n(a+b)\ge \Delta_n(a)+\Delta_n(b),\qquad
\|c\|^\prime_{\xx_n}+\|d\|^\prime_{\xx_n}\ge \|c+d\|^\prime_{\xx_n}.
$$
Therefore, we see that ${\mathcal C}$ is a convex set.
To show that $\mathcal C$ is closed, take a sequence $\{\varrho^m\}$ in ${\mathcal C}$ with the {\sf
X}-parts $X(a^m,c^m)$ which converges to $\varrho$ with the {\sf
X}-part $X(a,c)$. Take $s\in \calv_n^+$ with $\delta_n(s)=1$ such
that $\lan a,s\ran \le \Delta_n(a)+\varepsilon$. Put
$\alpha=\varepsilon/\|s\|_1$. Then, for every $m$ with
$\|a^m-a\|_\infty <\alpha$, we have
$$
\|c^m\|^\prime_{\xx_n}\le \Delta_n(a^m)\le\Delta_n(a+\alpha{\bf 1})
\le \lan a+\alpha{\bf 1}, s\ran \le\lan a,s\ran +\alpha\|s\|_1
\le\Delta_n(a)+2\epsilon.
$$
Therefore, we see that $\|c\|^\prime_{\xx_n}\le \Delta_n(a)+2\varepsilon$.
Because $\varepsilon>0$ was arbitrary, we conclude that $\varrho$ belongs to ${\mathcal C}$.
\end{proof}

It seems to be well known among specialists that every self-adjoint matrix in the tensor product
can be expressed as a linear combination of tensor products of self-adjoint matrices.
See \cite{choi-effros-inject} for example. Indeed,
every self-adjoint element $z=\sum_{k=1}^n v_k\ot w_k\in (V\ot W)_\sa$ in the tensor product $V\ot W$ of
$*$-vector spaces $V$ and $W$ can be written by
$$
\begin{aligned}
z&=\frac 12\sum_{k=1}^n v_k\ot w_k + \frac 12\sum_{k=1}^n v_k^*\ot w_k^*\\
&=\sum_{k=1}^n\left(\frac{v_k+v_k^*}2\right)\ot \left(\frac{w_k+w_k^*}2\right)
-\sum_{k=1}^n\left(\frac{v_k-v_k^*}{2{\rm i}}\right)\ot \left(\frac{w_k-w_k^*}{2{\rm i}}\right),
\end{aligned}
$$
which belongs to tensor product $V_\sa\ot W_\sa$ of self-adjoint parts.
Now, we denote by ${\mathcal D}$, ${\mathcal T}$ and ${\mathcal S}$
the convex sets of all states, PPT states and separable states, respectively.
Then the above decomposition shows that both ${\mathcal S}$ and ${\mathcal D}$ generate the affine manifold
consisting of all self-adjoint matrices with trace one.
Therefore, we can apply Proposition \ref{general}
for two convex sets ${\mathcal C}\cap{\mathcal T}$ and ${\mathcal T}$,
since ${\mathcal S}\subset {\mathcal C}\cap{\mathcal T}\subset {\mathcal T}\subset {\mathcal D}$.

Take $c\in\calv_n^\sa$ so that $|c_\bfi|=1$ for each $\bfi\in I_{[n]}$, and consider the following $n$-qubit self-adjoint matrix
\begin{equation}\label{bkiughjh}
\varrho_t= 2^{-n} X({\bf 1},tc)
\end{equation}
for $t\ge 0$. We have the following:
\begin{itemize}
\item
$\varrho_t\in{\mathcal D}\ \Longleftrightarrow\ \varrho_t\in{\mathcal T} \Longleftrightarrow\ t\le 1$;
\item
$\varrho_t\in{\mathcal S}\ \Longleftrightarrow\ \varrho_t\in{\mathcal C}\cap{\mathcal T}\
\Longleftrightarrow\ \Delta_n({\bf 1})\ge \|tc\|^\prime_{\xx_n}\ \Longleftrightarrow\ t\le 1/\|c\|^\prime_{\xx_n}$.
\end{itemize}
Take $t_0=1/\|c\|^\prime_{\xx_n}$. We have $t_0<1$ by Theorem
\ref{sep-ppt} whenever we take $c\in\calv_n^\sa$ with $|c_\bfi|=1$
and nontrivial phase difference. Then we see that $\varrho_{t_0}$ is
a boundary point of ${\mathcal C}\cap{\mathcal T}$ and an interior
point of ${\mathcal T}$. Therefore, we conclude that ${\mathcal
T}\setminus{\mathcal C}\cap{\mathcal T}$ has the nonempty interior.
This tells us that Theorem \ref{criterion} detects nonzero volume of
PPT entanglement.

We also see that $\varrho_{t_0}$ belongs to $\partial{\mathcal S}\cap\inte{\mathcal T}$. Therefore,
$\varrho_{t_0}$ is an $n$-qubit boundary separable state with full ranks in the sense of \cite{chen_dj_boundary_sep},
because $\varrho\in \inte{\mathcal T}$ if and only if all the partial transposes of $\varrho$ have the full ranks.
Such states have been constructed in $3\otimes 3$ system \cite{ha-kye-sep-face}, $2\otimes 4$ system \cite{ha-kye_2x4}
and three qubit system \cite{{kye_multi_dual},{kye-viet}}.
We recall that a nontrivial phase difference occurs
only when $n\ge 3$, because dimension of $\calv_n^\ph\ominus \image\Theta_n$ is $2^{n-1}-n=0$ for $n=1,2$.

\section{Conclusion}

In this paper, we have defined two numbers $\Delta_n(a)$ and $\|c\|^\prime_{\xx_n}$ arising from
the diagonal part $a\in\calv_n^+$ and the anti-diagonal part $c\in\calv_n^\sa$ of an $n$-qubit {\sf X}-state $X(a,c)$, and showed that
$\varrho=X(a,c)$ is separable if and only if the inequality
$$
\Delta_n(a)\ge \|c\|^\prime_{\xx_n}
$$
holds. Since the {\sf X}-part of a separable $n$-qubit state is again separable, this inequality gives rise to
a necessary criterion for separability, which detects PPT entanglement of nonzero volume.
Two numbers $\Delta_n(a)$ and $\|c\|^\prime_{\xx_n}$ in the above inequality have natural upper and lower bounds
\begin{equation}\label{sep_PPT}
\min_{\bfi\in I_{[n]}}\{ \sqrt{a_{\bfi}a_{\bar\bfi}} \} \ge \Delta_n(a),\qquad
\|c\|^\prime_{\xx_n}\ge \|c\|_\infty,
\end{equation}
respectively, by (\ref{delta=esi}) and (\ref{dual-norm-ineq}). Note that $\varrho=X(a,c)$ is of PPT if and only if
the inequality $\min \{\sqrt{a_{\bfi}a_{\bar\bfi}}\}\ge \|c\|_\infty$
holds, and so the strict inequalities in (\ref{sep_PPT}) reflect the existence of PPT entanglement.

In order to estimate the above two numbers $\Delta_n(a)$ and $\|c\|^\prime_{\xx_n}$ more precisely,
we have introduced the notion of irreducible balanced multisets
of indices, and defined the number $\tilde\Delta_n(a)$ which is an upper bound for $\Delta_n(a)$. The number
$\tilde\Delta_n(a)$ can be easily evaluated with the entries of $a$, and actually coincides with $\Delta_n(a)$
for $n=1,2,3$.
It seems to be very difficult to evaluate the number $\delta_n(s)$ in terms of the entries,
even for the case of $n=3$. In this regard, it is remarkable that its \lq dual\rq\ object $\Delta_3(a)$ can be
evaluated by the identity $\Delta_3(a)=\tilde\Delta_3(a)$.
It would be very interesting to know if the identity $\Delta_n(a)=\tilde\Delta_n(a)$ holds for $n\ge 4$.

The norm  $\|u\|_{\xx_n}$ and its dual norm $\|c\|^\prime_{\xx_n}$
depend on the phase parts of $u\in\calv_n^\sa$ and $c\in \calv_n^\sa$ as well as the magnitude parts.
We have introduced the notions of phase identities and phase differences to explain these phenomena.
Nontrivial phase differences appear only when $n\ge 3$, and this reflects the fact that there exists no
PPT entanglement in the two qubit system.
We gave a lower bound for the dual norm $\|c\|_{\xx_n}'$ to see that our criteria detects nonzero volume of PPT entanglement
whenever $n\ge 3$. Evaluation of $\|c\|^\prime_{\xx_n}$ in terms of entries seems to be very difficult, even though
it was possible in several cases of the three qubit system \cite{{chen_han_kye},{han_kye_GHZ},{han_kye_phase}}.
It would be nice to evaluate $\|c\|^\prime_{\xx_n}$ when the entries of $c\in\calv_n^\sa$ share a common magnitude.
See \cite{{chen_han_kye},{han_kye_phase}} for the formula in the three qubit case.

Note that there are other notions for separability like bi-separability and full bi-separability, according to bi-partitions of local systems.
Such notions of separability have been already characterized for {\sf X}-states in \cite{gao,guhne10,han_kye_optimal,Rafsanjani,SU}.
In these characterization, the phase part of anti-diagonal plays no role.


\end{document}